\PassOptionsToPackage{expansion=false,protrusion=true}{microtype}

\documentclass[sigconf, nonacm]{acmart}

\acmConference{}{}{}
\acmYear{}
\copyrightyear{}
\setcopyright{none}

\usepackage{amsmath}
\usepackage{algorithm}
\usepackage{algpseudocode}
\usepackage{booktabs}
\usepackage{subcaption}
\usepackage{xspace}
\usepackage{listings}
\usepackage{tikz}
\usepackage{pgfplots}
\usepackage{placeins}  
\usepackage{multirow}
\usetikzlibrary{arrows.meta, positioning, calc, patterns, shapes.geometric, decorations.pathreplacing, shadows}

\definecolor{surgePrimary}{HTML}{0077BB}
\definecolor{surgeLight}{HTML}{88CCEE}
\definecolor{baselineRed}{HTML}{CC3311}
\definecolor{baselineLight}{HTML}{EE6677}
\definecolor{compEncode}{HTML}{EE7733}
\definecolor{compSerialize}{HTML}{009988}
\definecolor{compUpload}{HTML}{AA3377}
\definecolor{compOverhead}{HTML}{BBBBBB}
\definecolor{accentGold}{HTML}{DDAA33}
\definecolor{archInput}{HTML}{DCEEF8}
\definecolor{archCore}{HTML}{FDECD0}
\definecolor{archOutput}{HTML}{E0F2E9}
\definecolor{archIO}{HTML}{F0E0F0}

\pgfplotsset{
  compat=1.18,
  every axis/.append style={
    font=\small,
    tick label style={font=\scriptsize},
    label style={font=\small},
    legend style={
      font=\scriptsize,
      draw=gray!40,
      fill=white,
      fill opacity=0.85,
      text opacity=1,
      rounded corners=1pt,
      inner sep=2pt,
      outer sep=1pt,
    },
    grid=major,
    grid style={gray!15, very thin},
    major tick length=2pt,
  },
  every axis plot/.append style={thick},
  surge bar/.style={
    ybar,
    bar width=8pt,
    enlarge x limits=0.12,
    nodes near coords style={font=\tiny, /pgf/number format/1000 sep={\,}},
    scaled y ticks=false,
    y tick label style={/pgf/number format/1000 sep={\,}},
  },
  surge line/.style={
    mark size=2.5pt,
    scaled y ticks=false,
    y tick label style={/pgf/number format/1000 sep={\,}},
  },
}

\newcommand{\surge}{\textsc{Surge}\xspace}
\newcommand{\bmin}{B_{\min}}
\newcommand{\bmax}{B_{\max}}
\newcommand{\nmax}{n_{\max}}
\newcommand{\cipc}{c_{\text{ipc}}}
\newcommand{\cenc}{c_{\text{enc}}}

\newtheorem{theorem}{Theorem}
\newtheorem{lemma}[theorem]{Lemma}
\newtheorem{corollary}[theorem]{Corollary}

\begin{document}

\title{SURGE: SuperBatch Unified Resource-efficient GPU Encoding\\for Heterogeneous Partitioned Data}

\author{Shashank Kapadia}
\affiliation{%
  \institution{Walmart Inc.}
  \country{USA}
}
\email{shashank.kapadia@walmart.com}

\author{Deep Narayan Mishra}
\affiliation{%
  \institution{Walmart Inc.}
  \country{USA}
}
\email{deep.mishra@walmart.com}

\author{Sujal Reddy Alugubelli}
\affiliation{%
  \institution{Walmart Inc.}
  \country{USA}
}
\email{SujalReddy.Alugubelli@samsclub.com}

\author{Ajay Kumar}
\affiliation{%
  \institution{Walmart Inc.}
  \country{USA}
}
\email{Ajay.Kumar10@walmart.com}

\author{Swapnil Yadav}
\affiliation{%
  \institution{Walmart Inc.}
  \country{USA}
}
\email{Swapnil.Yadav@walmart.com}

\author{Rishi Bhatia}
\affiliation{%
  \institution{Walmart Inc.}
  \country{USA}
}
\email{Rishi.Bhatia@walmart.com}

\renewcommand{\shortauthors}{Kapadia et al.}

\begin{abstract}
We present \textbf{SURGE}, a streaming GPU encoding system deployed in production to generate embeddings for over 800 million texts across 40{,}000 logical partitions. Production embedding pipelines face a tension between logical data partitioning and efficient GPU utilization: processing each partition independently incurs $P$ inter-process communication (IPC) calls whose overhead limits throughput for compute-light models. Our central contributions are analytical: (i)~a cost model (Theorem~\ref{thm:ipc}) that quantifies when IPC amortization matters as a function of partition distribution and model compute intensity, predicting throughput within 2\% of measured results across three encoders spanning a $15\times$ parameter range; (ii)~a memory-safety bound (Lemma~\ref{lem:memory}) that enables a streaming two-threshold aggregation policy with peak memory $O(\bmin + \nmax)$ rather than $O(N)$; and (iii)~a $\phi/\text{CV}$ decision framework (\S\ref{sec:generalizability}) that characterizes when the pattern applies beyond our workload. The naive fix of ignoring partitions and batching at fixed size reaches the IPC ceiling but requires $O(N)$ peak memory (32.7\,GB at 10M texts; infeasible beyond ${\sim}$60M on 192\,GB nodes), produces no output until all encoding completes, and offers no fault tolerance. SURGE achieves the same throughput with $O(\bmin + \nmax)$ bounded memory (2.6\,GB), $68\times$ faster time-to-first-output through incremental partition flushing, and crash recovery at SuperBatch granularity. On 10M texts with 4 NVIDIA L4 GPUs, SURGE delivers 26{,}413\,texts/s---matching fixed-batch throughput while using $12.6\times$ less memory. We validate the cost model on bge-base (109M, $d{=}768$, error 1.3\%) and across log-normal $\sigma \in \{1.0, 1.72, 2.5\}$ (speedup invariant within $\pm$3\%), and compare against a partition-batched baseline with offline columnar load balancing (PB-PBP-LB), against which SURGE retains a $7\%$ throughput edge, $2.5\times$ faster TTFO, and an unconditional memory guarantee. Complementary engineering---zero-copy Arrow serialization ($22$--$25\times$ over naive construction) and asynchronous I/O pipelining (up to $93\%$ benefit at high storage latency)---realizes the design but is not the contribution.
\end{abstract}


\maketitle

\section{Introduction}\label{sec:intro}

Dense text embeddings are a foundational building block for information retrieval~\cite{karpukhin2020dense}, recommendation systems~\cite{naumov2019dlrm}, and semantic search~\cite{devlin2019bert, reimers2019sentence}. A modern product catalog may contain hundreds of millions of text descriptions that must be encoded into fixed-dimensional vectors, stored in partitioned columnar formats, and refreshed regularly. The computational cost of this encoding step is dominated by GPU inference, and the engineering challenge lies in maximizing throughput across a workload whose inherent structure---logical partitions defined by domain taxonomy---conflicts with the batch-size requirements of efficient GPU execution.

The central insight of this paper is that \textbf{inter-process communication (IPC) overhead is a substantial throughput bottleneck} for production embedding pipelines using compute-light models. Multi-GPU encoding frameworks such as Sentence-Transformers~\cite{reimers2020multilingual} distribute work across processes via data serialization, transfer, and result gathering. Each invocation incurs a fixed IPC cost $\cipc$ regardless of payload size. When a pipeline processes $P = 4{,}000$ partitions independently, it makes 4{,}000 encoding calls; the aggregate IPC cost accounts for nearly half (48\%) of total wall-clock time in our evaluated workload (\S\ref{sec:eval}).

We formalize this bottleneck through a cost model (Theorem~\ref{thm:ipc}) and show that \emph{any} batching strategy reducing encoding invocations from $P$ to $O(N/B)$ achieves the same throughput ceiling. Current approaches fall into two categories:
\begin{itemize}
    \item \textbf{Partition-by-partition (PBP)} processing encodes each partition independently, incurring $P$ IPC calls. Throughput degrades proportionally to partition count.
    \item \textbf{Fixed-size batching (FSB)} ignores partition boundaries, encoding texts in chunks of predetermined size. This reduces IPC calls to $\lceil N / B \rceil$ and achieves the throughput ceiling, but requires $O(N)$ peak memory (32.7\,GB at 10M texts; Table~\ref{tab:main-results}) for an $O(N \log N)$ regrouping pass, produces zero output until all encoding completes, and loses all progress on failure.
\end{itemize}

\newpage
\noindent\textbf{Contributions.} The paper's primary contributions are analytical, not algorithmic:
\begin{enumerate}
    \item \textbf{Cost model (Theorem~\ref{thm:ipc}).} A first-principles decomposition of wall-clock time into IPC and compute components that predicts throughput for any batching strategy within $2\%$ of measured values, validated on three encoders (MiniLM 22M, bge-base 109M, E5-large 335M) and across log-normal $\sigma \in \{1.0, 1.72, 2.5\}$.
    \item \textbf{Memory-safety bound (Lemma~\ref{lem:memory}) and the two-threshold streaming policy.} An unconditional $O(\bmin + \nmax)$ peak-memory bound where $\nmax$ is the largest single-partition size, enabling a streaming aggregator with a lower efficiency threshold $\bmin$ and an upper safety trigger $\bmax$. This bound holds for adversarial partition arrival orders---the property that distinguishes SURGE from offline-sort baselines (\S\ref{sec:pbpbp-comparison}).
    \item \textbf{$\phi/\text{CV}$ decision framework (\S\ref{sec:generalizability}).} A two-parameter rule using the IPC-dominated fraction $\phi$ and the partition-size coefficient of variation that characterizes when the SURGE pattern is beneficial, applicable, or unnecessary---generalizing beyond our workload to multilingual corpora, geo-partitioned datasets, and any taxonomy-organized catalog.
\end{enumerate}

\noindent\textbf{System realization.} We embody these contributions in \surge, a streaming aggregation system. Three complementary engineering techniques realize the design but are not the contribution: (a)~a SuperBatch aggregator that accumulates completed partitions and flushes when a running text count crosses $\bmin$, with partition boundaries enabling zero-overhead slicing of the resulting embedding matrix (\S\ref{sec:aggregator}); (b)~zero-copy Arrow serialization that converts embedding matrices directly into FixedSizeListArrays, eliminating $O(Nd)$ intermediate Python objects (\S\ref{sec:zerocopy}); (c)~an asynchronous I/O pipeline that overlaps serialization and storage writes with GPU computation (\S\ref{sec:async}). Each is well-established in isolation; the contribution is their principled combination under the cost model and memory bound, and the validated deployment of the resulting system in production for $800$M$+$ texts.

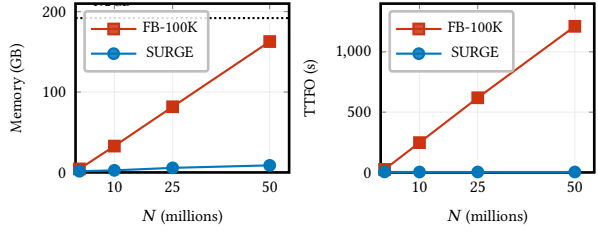
\begin{figure}[t]
\centering
\begin{tikzpicture}
\begin{axis}[
    width=0.52\columnwidth, height=0.45\columnwidth,
    xlabel={$N$ (millions)}, ylabel={Memory (GB)},
    xmin=0, xmax=55, ymin=0, ymax=210,
    xtick={10,25,50}, ytick={0,100,200},
    tick label style={font=\scriptsize},
    label style={font=\scriptsize},
    legend style={at={(0.03,0.97)}, anchor=north west, font=\scriptsize, draw=gray!50},
    mark size=2pt,
    line width=1pt,
]
\addplot[mark=square*, baselineRed, thick] coordinates {(1,4.4)(10,32.8)(25,81.7)(50,162.8)};
\addlegendentry{FB-100K}
\addplot[mark=*, surgePrimary, thick] coordinates {(1,1.5)(10,2.6)(25,5.7)(50,8.8)};
\addlegendentry{SURGE}
\addplot[mark=none, black, densely dotted, thick] coordinates {(0,192)(55,192)};
\node[font=\tiny, anchor=south west] at (axis cs:2,195) {192\,GB*};
\end{axis}
\end{tikzpicture}%
\hspace{2pt}%
\begin{tikzpicture}
\begin{axis}[
    width=0.52\columnwidth, height=0.45\columnwidth,
    xlabel={$N$ (millions)}, ylabel={TTFO (s)},
    xmin=0, xmax=55, ymin=0, ymax=1400,
    xtick={10,25,50}, ytick={0,500,1000},
    tick label style={font=\scriptsize},
    label style={font=\scriptsize},
    legend style={at={(0.03,0.97)}, anchor=north west, font=\scriptsize, draw=gray!50},
    mark size=2pt,
    line width=1pt,
]
\addplot[mark=square*, baselineRed, thick] coordinates {(1,25)(10,247.3)(25,619.3)(50,1211.8)};
\addlegendentry{FB-100K}
\addplot[mark=*, surgePrimary, thick] coordinates {(1,4.5)(10,3.6)(25,3.6)(50,3.6)};
\addlegendentry{SURGE}
\end{axis}
\end{tikzpicture}
\caption{SURGE's key advantage: bounded $O(\bmin + \nmax)$ memory and $O(1)$ time-to-first-output (TTFO) vs.\ fixed-batch's $O(N)$ scaling. At 50M texts, SURGE uses $18.5\times$ less memory and produces output $337\times$ faster. SURGE's TTFO decreases from 4.5\,s at 1M to 3.6\,s at 10M+ because model warmup and process initialization are amortized over larger first SuperBatches. Full analysis in \S\ref{sec:scaling}.}
\label{fig:teaser}
\end{figure}

\section{Background and Motivation}\label{sec:background}

\subsection{The Heterogeneous Partition Problem}\label{sec:partition-problem}

In production embedding pipelines, data is logically partitioned by domain taxonomies (product categories, language codes, geographic regions). These partitions exhibit heavy-tailed size distributions: a small number of large partitions contain a disproportionate share of texts, while the majority are small. In our workload, partition sizes follow a log-normal distribution ($\mu = 9.03$, $\sigma = 1.72$ in log-space), with sizes ranging from 187 to 447{,}231 texts and a median of 8{,}412. Figure~\ref{fig:partition-dist} illustrates this distribution.

\begin{figure}[t]
\centering
\begin{tikzpicture}
\begin{axis}[
    width=\columnwidth,
    height=0.618\columnwidth,  
    xlabel={Partition size (texts, log scale)},
    ylabel={Count},
    xmode=log,
    ybar,
    bar width=8pt,
    ymin=0,
    xtick={100, 1000, 10000, 100000},
    xticklabels={$10^2$, $10^3$, $10^4$, $10^5$},
    legend style={at={(0.97,0.97)}, anchor=north east, font=\scriptsize},
    title style={font=\small},
]
\fill[baselineRed!6] (axis cs:100,0) rectangle (axis cs:2340,860);
\node[font=\tiny, text=baselineRed!70!black, rotate=90, anchor=south]
    at (axis cs:500, 430) {IPC-dominated};
\addplot[fill=surgePrimary!50, draw=surgePrimary!80] coordinates {
    (200, 120) (400, 280) (800, 520) (1500, 710) (3000, 850)
    (6000, 620) (12000, 400) (25000, 230) (50000, 120)
    (100000, 80) (200000, 50) (400000, 20)
};
\draw[accentGold, dashed, line width=1.2pt] (axis cs:2340,0) -- (axis cs:2340,920);
\node[font=\scriptsize, text=accentGold!80!black, anchor=south] at (axis cs:2340,930) {$n^*$};
\end{axis}
\end{tikzpicture}
\caption{Partition size distribution (log-normal, $\mu{=}9.03$, $\sigma{=}1.72$). The dashed line marks $n^*{=}2{,}340$, the IPC-dominated threshold: 23\% of partitions ($\phi{=}0.23$) fall below $n^*$. However, the aggregate IPC cost across all $P{=}4{,}000$ calls accounts for 48\% of PBP wall time.}
\label{fig:partition-dist}
\end{figure}
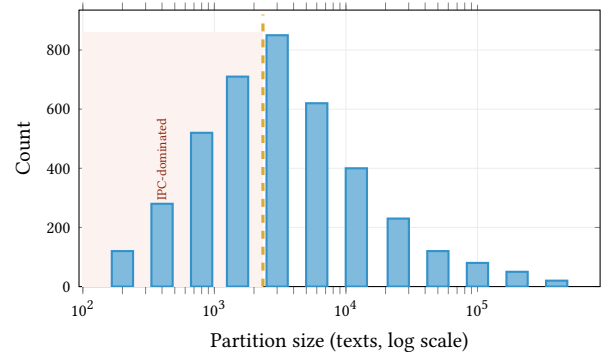

\subsection{Baseline Approaches: PBP and FSB}\label{sec:baselines-formal}

We formally define the two baseline strategies before introducing the cost model. Let the input be a set of partitions $\{(k_1, T_1), \ldots, (k_P, T_P)\}$ where $T_i$ is a sequence of $n_i$ texts and $\sum_i n_i = N$. The output requirement is a per-partition mapping from $k_i$ to its embedding matrix $\mathbf{E}_i \in \mathbb{R}^{n_i \times d}$.

\textbf{Partition-by-partition (PBP).} For each $i \in [P]$ in arrival order, invoke the multi-GPU encoder once on $T_i$ to obtain $\mathbf{E}_i$, then write $(k_i, \mathbf{E}_i)$ to storage. PBP makes exactly $P$ encode calls and produces output incrementally; its peak memory is $O(\nmax \cdot d)$ where $\nmax = \max_i n_i$.

\textbf{Fixed-size batching (FSB).} Concatenate all texts $T = T_1 \,\Vert\, \ldots \,\Vert\, T_P$ together with a parallel array of partition labels. Encode $T$ in chunks of fixed size $B$, yielding $\lceil N/B \rceil$ encode calls. After encoding completes, regroup the output rows by partition label (an $O(N \log N)$ argsort pass) to produce per-partition matrices. FSB achieves the IPC-amortized throughput ceiling but requires $O(N)$ peak memory to hold the full embedding matrix prior to regrouping, and emits zero output until the final regrouping step.

The two strategies bracket the design space: PBP minimizes memory and TTFO but maximizes IPC overhead; FSB amortizes IPC but sacrifices memory bound and streaming output. SURGE (\S\ref{sec:design}) achieves the FSB throughput ceiling with PBP's deployability properties.

\subsection{Multi-GPU Encoding Cost Model}\label{sec:cost-model}

Let $G$ denote the number of GPUs. A single call to a multi-GPU encoding function (e.g., \texttt{encode\_multi\_process}~\cite{reimers2020multilingual}) incurs:
\begin{itemize}
    \item A \emph{fixed IPC overhead} $\cipc$: process pool dispatch, data serialization to worker processes, result gathering, and deserialization. This cost is approximately constant regardless of payload size.
    \item A \emph{per-text cost} $\cenc / G$: tokenization, GPU forward pass, and result transfer, divided across $G$ workers.
\end{itemize}

\noindent The wall-clock time for processing a partition of $n_k$ texts independently is:
\begin{equation}\label{eq:partition-time}
T_k = \cipc + \frac{n_k \cdot \cenc}{G}
\end{equation}

\noindent The per-partition throughput is $\tau_k = n_k / T_k$. For small partitions where $n_k \cdot \cenc / G \ll \cipc$, throughput degrades to $\tau_k \approx n_k / \cipc$---proportional to partition size rather than GPU capacity. We define the \emph{IPC-dominated threshold}:
\begin{equation}\label{eq:nstar}
n^* = \frac{\cipc \cdot G}{\cenc}
\end{equation}

\noindent Partitions with $n_k < n^*$ spend more time on IPC than encoding. The \emph{IPC-dominated fraction} $\phi = |\{k : n_k < n^*\}| / P$ quantifies the workload's susceptibility to this overhead. In our workload, $n^* \approx 2{,}340$ and $\phi = 0.23$: while only 23\% of partitions are individually IPC-dominated, the \emph{modeled} aggregate IPC component $P \cdot \cipc = 4{,}000 \times 0.087\text{\,s} = 348$\,s accounts for 48\% of total PBP wall time (Equation~\ref{eq:partition-time} summed over all partitions). This is not a direct measurement residual but the cost-model prediction for the IPC portion of each encode call, which we validate against measured throughput in \S\ref{sec:e2e-results}. This aggregate overhead motivates batching across partition boundaries.

A natural alternative is eliminating IPC entirely via shared-memory multi-GPU approaches (e.g., \texttt{torch.nn.DataParallel}, thread-based workers). However, Python's Global Interpreter Lock (GIL) prevents true parallel tokenization across threads, and shared-memory GPU access from multiple threads requires careful synchronization. Process-based isolation provides fault containment (a GPU out-of-memory error in one worker does not crash the pipeline), independent CUDA context management, and clean memory accounting---properties essential for 6+ hour production runs processing 800M+ texts.

SURGE's amortization approach is complementary: it reduces IPC calls from $O(P)$ to $O(N/\bmin)$ within the existing process-based architecture and provides additional benefits (memory bounding, streaming output, crash recovery) that persist even if IPC overhead were eliminated.

\subsection{Why GPU Utilization Is Low}\label{sec:gpu-util-motivation}

A natural question is whether low throughput implies low GPU utilization. The answer depends on the \emph{compute intensity} of the model. For a compute-light model such as MiniLM-L6-v2 (22M parameters)~\cite{wang2020minilm}, the GPU forward pass is fast relative to tokenization and data movement. Even within an encode call operating at full batch size, the GPU is idle during tokenization and CPU-GPU data transfer. GPU utilization measured by hardware counters (e.g., \texttt{nvidia-smi}) reflects kernel occupancy, not pipeline throughput---the observed ${\sim}10\%$ utilization under SURGE (Table~\ref{tab:main-results}) is a consequence of model size, not system inefficiency. We formalize this observation in \S\ref{sec:compute-intensity}.

\section{System Design}\label{sec:design}

\subsection{System Overview}

\surge consists of five components in a streaming pipeline: (1)~an ordered data source providing rows sorted by partition key, (2)~partition boundary detection via key-change monitoring, (3)~a SuperBatch aggregator that accumulates partitions for GPU-efficient encoding, (4)~zero-copy serialization to columnar format, and (5)~asynchronous storage upload. Figure~\ref{fig:architecture} illustrates the pipeline.

\begin{figure}[t]
\centering
\begin{tikzpicture}[
    node distance=0.15cm and 0.22cm,
    stage/.style={draw=gray!60, rounded corners=4pt, minimum height=0.72cm, minimum width=1.12cm,
                  font=\scriptsize, align=center, thick,
                  drop shadow={shadow xshift=0.4pt, shadow yshift=-0.4pt, opacity=0.12}},
    arr/.style={-{Stealth[length=4pt]}, thick, gray!70},
]
\node[stage, fill=archInput] (src) {Ordered\\Source};
\node[stage, fill=archInput, right=of src] (bnd) {Boundary\\Detect};
\node[stage, fill=archCore, right=of bnd, draw=compEncode!80, line width=1pt] (agg) {SuperBatch\\Aggregator};
\node[stage, fill=archOutput, right=of agg] (ser) {Zero-Copy\\Serialize};
\node[stage, fill=archIO, right=of ser] (upl) {Async\\Upload};

\draw[arr] (src) -- (bnd);
\draw[arr] (bnd) -- (agg);
\draw[arr] (agg) -- (ser);
\draw[arr] (ser) -- (upl);

\begin{scope}[yshift=-1.35cm, xshift=0.6cm]
  \node[font=\tiny\bfseries, anchor=east] at (-0.1, 0.12) {Time:};
  \fill[compEncode!40] (0, 0) rectangle (1.6, 0.28);
  \fill[compSerialize!40] (1.6, 0) rectangle (2.05, 0.28);
  \fill[compUpload!40] (2.05, 0) rectangle (2.85, 0.28);
  \node[font=\tiny, white] at (0.8, 0.14) {GPU $j$};
  \fill[compEncode!60] (1.6, -0.38) rectangle (3.2, -0.1);
  \node[font=\tiny, white] at (2.4, -0.24) {GPU $j{+}1$};
  \draw[decorate, decoration={brace, amplitude=2pt, mirror}, accentGold, thick]
      (1.6, -0.52) -- (2.85, -0.52) node[midway, below=2pt, font=\tiny, accentGold!80!black] {overlap};
  \node[font=\tiny, anchor=west] at (3.4, 0.0) {
    \tikz{\fill[compEncode!50] (0,0) rectangle (0.18,0.12);} Enc
    \tikz{\fill[compSerialize!40] (0,0) rectangle (0.18,0.12);} Ser
    \tikz{\fill[compUpload!40] (0,0) rectangle (0.18,0.12);} Upl
  };
\end{scope}
\end{tikzpicture}
\caption{\surge pipeline architecture. GPU encoding of SuperBatch $j{+}1$ overlaps with serialization and upload of SuperBatch $j$, eliminating I/O stalls. The mini-timeline shows how pipelining hides I/O latency.}
\label{fig:architecture}
\end{figure}
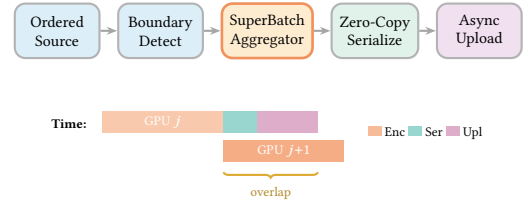

\subsection{SuperBatch Aggregator}\label{sec:aggregator}

The SuperBatchAggregator bridges the gap between logical partitions and physical GPU batches. Algorithm~\ref{alg:surge} presents the core logic; its peak resident state is bounded by $O(\bmin + \nmax)$ where $\nmax$ is the largest single-partition size in the input (Lemma~\ref{lem:memory}).

\begin{algorithm}[t]
\caption{SURGE SuperBatch Aggregation. Peak resident state $O(\bmin + \nmax)$.}
\label{alg:surge}
\begin{algorithmic}[1]
\Require Stream of $(key, text)$ pairs ordered by key; thresholds $\bmin$, $\bmax$; $G$ GPUs; model $f_\theta$
\Ensure Per-partition columnar files in remote storage
\State $\textit{partitions} \gets []$; $\textit{total} \gets 0$; $\textit{curKey} \gets \texttt{null}$; $\textit{curTexts} \gets []$
\For{each $(key, text)$ in stream}
    \If{$key \neq curKey$}
        \If{$curKey \neq \texttt{null}$}
            \State \Call{AddPartition}{$curKey$, $curTexts$}
        \EndIf
        \State $curKey \gets key$; $curTexts \gets []$
    \EndIf
    \State $curTexts.\text{append}(text)$
\EndFor
\State \Call{AddPartition}{$curKey$, $curTexts$}; \Call{Flush}{}
\Statex
\Procedure{AddPartition}{$key$, $texts$}
    \State $partitions.\text{append}((key, \text{copy}(texts)))$
    \State $total \gets total + |texts|$
    \If{$total \geq \bmax$} \Call{Flush}{} \Comment{Memory-safety trigger (rare)}
    \ElsIf{$total \geq \bmin$} \Call{Flush}{} \Comment{Efficiency trigger (common)}
    \EndIf
\EndProcedure
\Statex
\Procedure{Flush}{}
    \State $allTexts \gets []$; $bounds \gets []$; $idx \gets 0$
    \For{each $(key, texts)$ in $partitions$}
        \State $allTexts.\text{extend}(texts)$
        \State $bounds.\text{append}((idx, idx + |texts|, key))$
        \State $idx \gets idx + |texts|$
    \EndFor
    \State $\mathbf{E} \gets f_\theta.\text{encode\_multi\_process}(allTexts)$ \Comment{Single GPU call}
    \For{each $(start, end, key)$ in $bounds$}
        \State $\mathbf{E}_k \gets \mathbf{E}[start{:}end]$ \Comment{Zero-copy slice}
        \State $bytes \gets$ \Call{ZeroCopySerialize}{$key$, $allTexts[start{:}end]$, $\mathbf{E}_k$}
        \State \Call{AsyncUpload}{$key$, $bytes$}
    \EndFor
    \State $partitions \gets []$; $total \gets 0$
\EndProcedure
\end{algorithmic}
\end{algorithm}

\textbf{Input ordering.} Algorithm~\ref{alg:surge} assumes input is ordered by partition key (line~1). For streaming sources where data may arrive out-of-order, a pre-sort pass or partition-key grouping stage is required before SURGE ingestion. In the worst case this is $O(N \log N)$---the same complexity attributed to fixed-batch's regrouping pass. However, in practice, partitioned data stores (Hive, BigQuery, Spark DataFrames) return rows grouped by partition key natively, satisfying this requirement without additional processing. In our production deployment, Spark's hash partitioner provides the ordering guarantee at data read time.

\textbf{Design decisions.} The \texttt{copy(texts)} on line~12 of Algorithm~\ref{alg:surge} snapshots the current partition's text list before the caller clears it for the next partition; this is a shallow list copy ($O(n_k)$ reference copies), not a deep string copy, and ensures correctness without measurable overhead.

\textbf{The two-threshold scheme.} The lower threshold $\bmin = 100{,}000$ is the \emph{efficiency trigger}: it ensures each GPU call processes at least $\bmin / G$ texts per GPU, well above the IPC-dominated regime, and is hit on the common path. The upper threshold $\bmax = 500{,}000$ is the \emph{memory-safety trigger}: it is a rare-but-required emergency flush that fires only when a single arriving partition would push the running total past $\bmax$ before the next $\bmin$ flush would otherwise occur. $\bmax$ is not a second efficiency knob---it bounds peak memory unconditionally, including under adversarial arrival orders where one large partition immediately follows a SuperBatch already near $\bmin$. With $\bmax = 500{,}000$, $L = 47$, $d = 384$, the data-resident memory bound is:
\begin{equation}\label{eq:memory}
M(\bmax) = \bmax \cdot (L + 4d) = 791\text{\,MB},
\end{equation}
within a single GPU's 24\,GB VRAM budget. No partition is split across SuperBatches under normal operation; the rare oversized-partition case ($n_k > \bmax$) is handled by emitting the partition as its own SuperBatch (\S\ref{sec:operational}).

\subsection{Asynchronous I/O Pipeline}\label{sec:async}

A synchronous baseline encodes a SuperBatch, then serializes and uploads its outputs before issuing the next encode call. With per-batch encode time $t_{\text{enc}}$ and combined I/O time $t_{\text{ser}} + t_{\text{upl}}$, total wall time is $\sum_j (t_{\text{enc},j} + t_{\text{ser},j} + t_{\text{upl},j})$---I/O latency stacks linearly on the critical path. The async I/O pipeline decouples GPU computation from storage writes using a producer-consumer pattern backed by a thread pool (Algorithm~\ref{alg:async}), enabling I/O of batch $j$ to overlap with the encode of batch $j{+}1$.

\begin{algorithm}[t]
\caption{Asynchronous Storage Upload}
\label{alg:async}
\begin{algorithmic}[1]
\Require Thread pool with $W$ workers; retry: max 3 attempts, $2^{\text{attempt}}$\,s backoff
\State $pending \gets \{\}$
\Procedure{AsyncUpload}{$path$, $data$}
    \State $future \gets \text{pool.submit}(\textsc{UploadWithRetry}, path, data)$
    \State $pending[path] \gets future$ \Comment{Non-blocking return}
\EndProcedure
\Procedure{UploadWithRetry}{$path$, $data$}
    \For{$a = 0$ to $2$}
        \State \textbf{try} $\text{storage.write}(path, data)$; \textbf{return}
        \State \textbf{catch} \textbf{sleep}($2^a$\,s)
    \EndFor
\EndProcedure
\end{algorithmic}
\end{algorithm}

\textbf{Overlap model.} Let $t_{\text{enc}}$, $t_{\text{ser}}$, $t_{\text{upl}}$ denote the times for encoding, serialization, and upload of a single SuperBatch. Without pipelining, each SuperBatch takes $t_{\text{enc}} + t_{\text{ser}} + t_{\text{upl}}$. With async I/O, serialization and upload of batch $j$ overlap with encoding of batch $j{+}1$, reducing effective time per batch to $\max(t_{\text{enc}}, t_{\text{ser}} + t_{\text{upl}})$. We define the \emph{I/O overlap ratio}:
\begin{equation}\label{eq:overlap}
\rho = 1 - \frac{\max(0, (t_{\text{ser}} + t_{\text{upl}}) - t_{\text{enc}})}{t_{\text{ser}} + t_{\text{upl}}}
\end{equation}
When $t_{\text{enc}} \geq t_{\text{ser}} + t_{\text{upl}}$, $\rho = 1$ (perfect overlap). When storage latency is high enough that $t_{\text{ser}} + t_{\text{upl}} > t_{\text{enc}}$, asynchronous pipelining prevents this from stalling the GPU. Figure~\ref{fig:pipeline-overlap} illustrates the overlap.

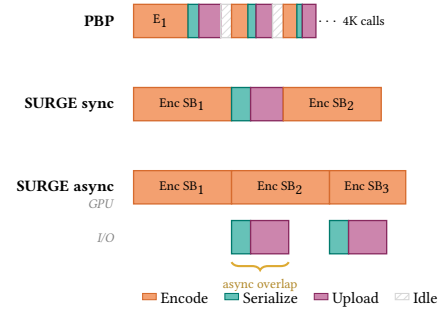
\begin{figure}[tbp]
\centering
\begin{tikzpicture}[
    x=0.72cm, y=0.55cm,
    enc/.style={fill=compEncode!70, draw=compEncode!90!black, line width=0.5pt},
    ser/.style={fill=compSerialize!60, draw=compSerialize!80!black, line width=0.5pt},
    upl/.style={fill=compUpload!60, draw=compUpload!80!black, line width=0.5pt},
    idle/.style={fill=compOverhead!30, draw=compOverhead!60, line width=0.5pt, pattern=north east lines, pattern color=compOverhead!50},
    lbl/.style={font=\tiny, text=black},
]
\node[anchor=east, font=\scriptsize\bfseries] at (-0.2, 5.9) {PBP};
\draw[enc] (0,5.5) rectangle (1.0,6.3) node[midway, lbl] {E$_1$};
\draw[ser] (1.0,5.5) rectangle (1.2,6.3);
\draw[upl] (1.2,5.5) rectangle (1.6,6.3);
\draw[idle] (1.6,5.5) rectangle (1.8,6.3);
\draw[enc] (1.8,5.5) rectangle (2.1,6.3);
\draw[ser] (2.1,5.5) rectangle (2.25,6.3);
\draw[upl] (2.25,5.5) rectangle (2.55,6.3);
\draw[idle] (2.55,5.5) rectangle (2.75,6.3);
\draw[enc] (2.75,5.5) rectangle (3.0,6.3);
\draw[ser] (3.0,5.5) rectangle (3.1,6.3);
\draw[upl] (3.1,5.5) rectangle (3.35,6.3);
\node[lbl] at (4.0,5.9) {$\cdots$ 4K calls};

\node[anchor=east, font=\scriptsize\bfseries] at (-0.2, 3.9) {SURGE sync};
\draw[enc] (0,3.5) rectangle (1.8,4.3) node[midway, lbl] {Enc SB$_1$};
\draw[ser] (1.8,3.5) rectangle (2.15,4.3);
\draw[upl] (2.15,3.5) rectangle (2.75,4.3);
\draw[enc] (2.75,3.5) rectangle (4.55,4.3) node[midway, lbl] {Enc SB$_2$};

\node[anchor=east, font=\scriptsize\bfseries] at (-0.2, 1.9) {SURGE async};
\node[anchor=east, font=\tiny, gray] at (-0.2, 1.5) {\textit{GPU}};
\draw[enc] (0,1.5) rectangle (1.8,2.3) node[midway, lbl] {Enc SB$_1$};
\draw[enc] (1.8,1.5) rectangle (3.6,2.3) node[midway, lbl] {Enc SB$_2$};
\draw[enc] (3.6,1.5) rectangle (5.0,2.3) node[midway, lbl] {Enc SB$_3$};
\node[anchor=east, font=\tiny, gray] at (-0.2, 0.7) {\textit{I/O}};
\draw[ser] (1.8,0.3) rectangle (2.15,1.1);
\draw[upl] (2.15,0.3) rectangle (2.85,1.1);
\draw[ser] (3.6,0.3) rectangle (3.95,1.1);
\draw[upl] (3.95,0.3) rectangle (4.65,1.1);

\draw[decorate, decoration={brace, amplitude=3pt, mirror}, accentGold, thick]
    (1.8, 0.1) -- (2.85, 0.1) node[midway, below=3pt, font=\tiny, accentGold!80!black] {async overlap};

\node[anchor=west, font=\scriptsize] at (0, -0.8) {
  \tikz[baseline=-0.5ex]{\fill[compEncode!70, draw=compEncode!90!black] (0,0) rectangle (0.25,0.15);} Encode\quad
  \tikz[baseline=-0.5ex]{\fill[compSerialize!60, draw=compSerialize!80!black] (0,0) rectangle (0.25,0.15);} Serialize\quad
  \tikz[baseline=-0.5ex]{\fill[compUpload!60, draw=compUpload!80!black] (0,0) rectangle (0.25,0.15);} Upload\quad
  \tikz[baseline=-0.5ex]{\fill[compOverhead!30, draw=compOverhead!60, pattern=north east lines, pattern color=compOverhead!50] (0,0) rectangle (0.25,0.15);} Idle
};
\end{tikzpicture}
\caption{Pipeline overlap. \textbf{PBP:} many small GPU calls with idle gaps between them. \textbf{SURGE sync:} few large encode calls but I/O blocks the next encode. \textbf{SURGE async:} I/O of SuperBatch $j$ overlaps with encode of SuperBatch $j{+}1$ on a separate I/O thread, eliminating storage stalls on the critical path.}
\label{fig:pipeline-overlap}
\end{figure}

\subsection{Zero-Copy Embedding Serialization}\label{sec:zerocopy}

The naive approach constructs $N \times d$ Python float objects:

\begin{lstlisting}[caption={Naive: $O(Nd)$ Python objects},label=lst:naive]
lists = [row.tolist() for row in emb]
table = pa.table({"embedding": lists})
\end{lstlisting}

\noindent For $N = 200{,}000$ and $d = 384$, this allocates ${\sim}$2.3\,GB of transient Python objects. Our zero-copy path:

\begin{lstlisting}[caption={Zero-copy: $O(1)$ allocations},label=lst:zerocopy]
flat = pa.array(emb.ravel(), type=pa.float32())
col = pa.FixedSizeListArray.from_arrays(flat, d)
table = pa.table({"embedding": col})
\end{lstlisting}

\noindent \texttt{ravel()} returns a view (zero-copy for C-contiguous arrays), \texttt{pa.array()} wraps the buffer without copying, and \texttt{FixedSizeListArray.\allowbreak{}from\_arrays()} records the list size with no data movement. The path allocates $O(1)$ Python objects regardless of $N$.

\textbf{Lifetime and aliasing.} Because the Arrow buffer aliases the NumPy array's memory, the embedding matrix must outlive every \texttt{AsyncUpload} future that references it: SURGE retains the matrix by capturing it in the upload closure, and the thread pool's completion of the future drops the last reference. The aliased buffer must not be mutated in place; downstream consumers receive a read-only view. In our pipeline this constraint is satisfied trivially because the output is write-once. Table~\ref{tab:serialization} quantifies the speedup.

\subsection{Multi-GPU Coordination}

The system supports 1 to $G$ GPUs via multi-process encoding. At startup, GPU type and count are detected via CUDA runtime queries and per-GPU batch sizes selected from a configuration table (T4/L4: 1{,}024; A100/H100: 2{,}048). The multi-GPU pool is started once and reused across all SuperBatch flushes, amortizing process spawn cost. CUDA optimizations include \texttt{cudnn.benchmark}, TF32 on Ampere+, and expandable memory segments to prevent fragmentation over long runs.

\subsection{Resume Capability}

Production pipelines must handle interruptions. The system implements idempotent resume by scanning the output storage prefix for existing partition paths. Each path is deterministic (constructed from partition key and run identifier), enabling an $O(P)$ existence check at startup. If failure occurs mid-SuperBatch, the entire SuperBatch is re-processed on resume---at most $\bmax$ texts are re-encoded, a bounded cost. Completed partitions from prior SuperBatches are skipped via the idempotent path check, ensuring exactly-once output semantics without a separate transaction log.

\section{Formal Analysis}\label{sec:analysis}

We develop a cost model that predicts SURGE's throughput improvement from first principles and validate it against measured results.

\subsection{IPC Amortization Bound}

\begin{theorem}[IPC Amortization]\label{thm:ipc}
Let $N$ texts be distributed across $P$ partitions, processed on $G$ GPUs with per-call IPC overhead $\cipc$ and per-text encoding cost $\cenc$. Define $\alpha = P \cdot \cipc / (N \cdot \cenc / G)$ as the \emph{IPC-to-compute ratio} for PBP processing, and let SURGE use threshold $\bmin$, producing $F = \lceil N / \bmin \rceil$ flushes. Then:
\begin{equation}\label{eq:speedup}
\text{Speedup} = \frac{T_{\textit{PBP}}}{T_{\textit{SURGE}}} = \frac{1 + \alpha}{1 + \alpha \cdot F / P}
\end{equation}
\end{theorem}

\begin{proof}
The PBP wall-clock time (ignoring I/O overlap) is:
\begin{equation}
T_{\textit{PBP}} = P \cdot \cipc + \frac{N \cdot \cenc}{G} = \frac{N \cdot \cenc}{G}(1 + \alpha)
\end{equation}
SURGE accumulates texts until reaching $\bmin$, then flushes. The number of encode calls is at most $F = \lceil N / \bmin \rceil$. Each call incurs one IPC overhead:
\begin{align}
T_{\textit{SURGE}} &= F \cdot \cipc + \frac{N \cdot \cenc}{G} \notag \\
&= \frac{N \cdot \cenc}{G}\!\left(1 + \alpha \cdot \frac{F}{P}\right)
\end{align}
The ratio yields Equation~\ref{eq:speedup}.
\end{proof}

\begin{corollary}[Regime analysis]\label{cor:regimes}
\hfill
\begin{enumerate}
    \item \textbf{IPC-dominated} ($\alpha \gg 1$): Speedup $\to P / F$. For $P = 4{,}000$ and $F = 100$: speedup $\to 40\times$.
    \item \textbf{Compute-dominated} ($\alpha \ll 1$): Speedup $\to 1$. SURGE provides no benefit when IPC is negligible.
    \item \textbf{Mixed regime} ($\alpha \approx 1$): The measured $\alpha$ determines the benefit. Estimating from our benchmark ($\cipc \approx 0.087$\,s, $\cenc \approx 0.149$\,ms, $G = 4$): $\alpha = 4{,}000 \times 0.087 / (10^7 \times 1.49 \times 10^{-4} / 4) = 348 / 372.5 = 0.93$. Predicted speedup $= (1 + 0.93) / (1 + 0.93 \times 100/4000) = 1.93 / 1.023 = 1.89$, closely matching the measured $1.92\times$ (Table~\ref{tab:main-results}), with error ${<}$2\%.
\end{enumerate}
\end{corollary}

\textbf{Cross-model validation.} Theorem~\ref{thm:ipc} predicts speedup as a function only of $\alpha$, $P$, and $F$---no model-specific tuning. We validate this on two additional encoders. On bge-base (109M, $d{=}768$, 2$\times$L4) the back-solved $\cipc{=}0.081$\,s and $\cenc{=}0.215$\,ms yield $\alpha{=}0.603$ and a predicted $1.31\times$ speedup over PBP; measured $1.29\times$ (error $1.3\%$; Table~\ref{tab:e5large}). On E5-large (335M, $d{=}1024$, 4$\times$L4) the predicted $1.34\times$ matches the measured $1.32\times$ (error $1.5\%$). Across three encoders spanning a $15\times$ parameter range, the model's prediction error stays within $2\%$, supporting Theorem~\ref{thm:ipc} as a workload-planning tool.

\subsection{Compute Intensity}\label{sec:compute-intensity}

We define the \emph{compute intensity} $\mathcal{I}$ as the ratio of GPU kernel time to total encode-call time:
\begin{equation}\label{eq:intensity}
\mathcal{I} = \frac{t_{\text{kernel}}}{t_{\text{tokenize}} + t_{\text{transfer}} + t_{\text{kernel}}}
\end{equation}
The expected GPU utilization under SURGE is bounded by the \emph{encode duty cycle} $\delta$ multiplied by compute intensity:
\begin{equation}\label{eq:gpu-util}
\text{GPU\%} \leq \delta \cdot \mathcal{I}, \quad \text{where } \delta = \frac{\sum_j t_{\text{enc},j}}{T_{\text{wall}}}
\end{equation}

For MiniLM-L6-v2 (22M parameters), $\mathcal{I}$ is low because tokenization and CPU-GPU transfer dominate the forward pass time for this small model. The encode duty cycle $\delta$ under SURGE is ${\sim}57\%$ (Table~\ref{tab:main-results}), and the product $\delta \cdot \mathcal{I}$ yields the observed ${\sim}10\%$ GPU utilization measured by \texttt{nvidia-smi}.

For larger models (e.g., E5-large at 335M parameters~\cite{wang2022e5}), $\mathcal{I}$ increases proportionally to model FLOPs while tokenization remains constant, predicting substantially higher GPU utilization. We validate this empirically: E5-large achieves 62\% GPU utilization under SURGE (vs.\ 10\% for MiniLM), confirming the compute intensity scaling. Quantitatively, for MiniLM under SURGE: $\delta = 57\%$ (Table~\ref{tab:main-results}), yielding $\mathcal{I} = 10.6\% / 57\% = 18.6\%$. For E5-large: $\delta \approx 85\%$ (higher encode fraction due to compute dominance), yielding $\mathcal{I} \approx 62\% / 85\% = 73\%$---a $3.9\times$ increase reflecting the $15\times$ larger model operating on the same tokenization and transfer pipeline. SURGE's core benefit remains: it amortizes IPC regardless of $\mathcal{I}$, and the memory/TTFO advantages persist (\S\ref{sec:threats}).

\subsection{Memory Bound}

\begin{lemma}[SuperBatch memory]\label{lem:memory}
The peak data-resident state of the SuperBatch aggregator is bounded by $O(\bmin + \nmax)$, where $\nmax$ is the largest single-partition size in the input. Concretely, for a SuperBatch holding $S$ texts with average length $L$ bytes and embedding dimension $d$:
\begin{equation}
M(S) = S \cdot L + S \cdot d \cdot 4 \text{ bytes}, \quad S \leq \bmin + \nmax \leq \bmax.
\end{equation}
For $S = \bmax = 500{,}000$, $d = 384$, $L = 47$: $M = 23.5\,\text{MB (text)} + 768\,\text{MB (embeddings)} = 791.5$\,MB.
\end{lemma}

\noindent The bound holds because Algorithm~\ref{alg:surge} flushes once $total \geq \bmin$, so the running buffer never exceeds the previous total ($< \bmin$) plus the latest partition's $n_k \leq \nmax$. The $\bmax$ trigger is a tighter unconditional ceiling that activates only when a single arriving partition would push the buffer past $\bmax$ before the $\bmin$ check fires---guaranteeing the bound even under adversarial arrival orders.

\noindent This data-resident lower bound accounts only for raw text and output embeddings; it is intentionally conservative. In practice, SURGE's measured peak memory is 2.5\,GB (Table~\ref{tab:main-results}), approximately $3\times$ the theoretical lower bound. The gap is attributable to three additional allocations: (1)~tokenizer buffers---token IDs and attention masks are \texttt{int64} tensors of shape $(\bmax, \text{seq\_len})$, adding ${\sim}500$\,MB; (2)~model-internal activations and intermediate tensors during encoding; and (3)~Python runtime overhead (interpreter, garbage collector, and SentenceTransformers process pool memory). These costs are constant across methods and independent of partition count, so the $O(\bmin + \nmax)$ scaling holds.

\subsection{Connection to Bin Packing}

SURGE's greedy accumulation can be viewed as a variant of the \emph{Next Fit} bin packing heuristic~\cite{coffman1978multiprocessor, coffman2013survey} with a minimum-fill constraint. Classical Next Fit opens a new bin when the current bin exceeds capacity $\bmax$; SURGE additionally requires the bin to reach $\bmin$ before closing. When partition sizes are i.i.d.\ with mean $\mu$ and variance $\sigma^2$, the expected fill ratio of each SuperBatch follows from renewal theory~\cite{feller1971introduction}. The SuperBatch accumulates partitions until the running sum of sizes first exceeds $\bmin$; by Wald's identity for random walks~\cite{wald1944sequential}, the expected overshoot is $\sigma^2 / (2\mu)$ for large $\bmin / \mu$, giving:
\begin{equation}\label{eq:fill-ratio}
\mathbb{E}\!\left[\frac{S}{\bmin}\right] = 1 + \frac{\sigma^2}{2\mu \bmin} + O\!\left(\frac{1}{\bmin}\right)
\end{equation}
For large $\bmin / \mu$ (i.e., many partitions per SuperBatch), this approaches 1---near-optimal packing. With $\mu = 8{,}412$, $\sigma = 17{,}660$, and $\bmin = 100{,}000$: the expected fill ratio is $1 + 17{,}660^2 / (2 \times 8{,}412 \times 100{,}000) \approx 1.19$, meaning SuperBatches are on average 19\% overfull relative to $\bmin$. This yields predictable flush counts and stable throughput. Offline alternatives such as First-Fit-Decreasing or Best-Fit-Decreasing achieve marginally tighter packing ($99.8\%$ vs.\ $99.2\%$ fill ratio in simulation) at the cost of requiring all partition sizes upfront---incompatible with the streaming arrival model that the $O(\bmin + \nmax)$ bound depends on.

\section{Evaluation}\label{sec:eval}

\subsection{Experimental Setup}\label{sec:setup}

\textbf{Dataset.} We generate a synthetic dataset of 10M texts across 4{,}000 partitions with sizes drawn from a log-normal distribution ($\mu = 9.03$, $\sigma = 1.72$) matching production workload characteristics (\S\ref{sec:partition-problem}). Texts are synthetic sentences averaging 47 bytes, consistent with product title lengths. Note: the production deployment processes 800M+ texts across 40{,}000 partitions; our benchmark uses a representative 10M-text subset (scaling to 50M in \S\ref{sec:scaling}, up to $P = 20{,}000$ partitions) to enable controlled evaluation.

\textbf{Hardware.} A single GCP g2-standard-48 node with 4 NVIDIA L4 GPUs (24\,GB VRAM each), 48 vCPUs, 192\,GB RAM. This configuration is denoted ``*'' in figures. Sub-experiments noted in their respective tables use 2$\times$L4 (g2-standard-24) for parity with model availability constraints during the bge-base and PB-PBP-LB runs.

\textbf{Model.} all-MiniLM-L6-v2~\cite{wang2020minilm}: 22M parameters, 384-dimensional L2-normalized embeddings. We use multi-process encoding via Sentence-Transformers~\cite{reimers2020multilingual} with batch size 1{,}024 per GPU. A warmup encode of 1{,}024 texts is performed before timing to initialize CUDA contexts. Generalization to bge-base-en-v1.5 (109M, $d{=}768$) and E5-large (335M, $d{=}1024$) is reported in \S\ref{sec:model-generalization}.

\textbf{Storage backend.} Results use a latency-simulated GCS backend (base latency 10\,ms, throughput 200\,MB/s per write) unless otherwise noted. This provides realistic I/O conditions while ensuring reproducibility.

\textbf{Baselines.} (1)~\emph{PBP}: each partition encoded independently. (2)~\emph{Fixed-batch (FSB)}: texts streamed in fixed chunks (10K, 50K, 100K) with an $O(N \log N)$ argsort-based regrouping pass. (3)~\emph{PB-PBP-LB}: a stronger partition-batched baseline that pre-sorts partitions by columnar size statistics and FFD-packs whole partitions into batches of size $B$ (\S\ref{sec:pbpbp-comparison}). (4)~\emph{SURGE (sync)}: SuperBatch aggregation without async I/O. (5)~\emph{SURGE + AsyncIO}: the full system. All methods use identical model weights, tokenizer settings, and output serialization format (Apache Arrow) to ensure fair comparison. The only variable is the batching and I/O strategy.

\textbf{Metrics.} Throughput (texts/s), GPU utilization (\%, \texttt{nvidia-smi} at 100\,ms via \texttt{pynvml}), encode duty cycle $\delta$ (fraction of wall time in encode calls), wall time (s), cost (\$/M texts at \$7.30/hr), I/O overlap ratio $\rho$ (Equation~\ref{eq:overlap}; $\rho{=}1$ indicates I/O fully overlaps with encode). Peak memory is measured via \texttt{psutil} RSS sampling at 100\,ms intervals; GPU memory is excluded as model weights are constant across methods.

\textbf{Reproducibility.} All results report the mean of 3 independent runs under controlled conditions (identical data, warmup, and cleanup between methods). Tables report mean $\pm$ 1 standard deviation where applicable. Run-to-run variance is ${<}$1\% for throughput metrics, confirming measurement stability. Between runs, we clear filesystem caches and restart GPU processes to eliminate warm-cache effects. The benchmark harness, synthetic data generator, and analysis scripts are included in the supplementary materials.

\subsection{End-to-End Results}\label{sec:e2e-results}

\begin{table*}[t]
\caption{End-to-end comparison on benchmark workload (GCS-profile storage, mean of 3 runs; run-to-run std $<$1\% for throughput). Duty\% = encode time / wall time $\times$ 100. Proposed method in \textbf{bold}.}
\label{tab:main-results}
\centering
\small
\begin{tabular}{lrrrrrcrc}
\toprule
Method & Tput (t/s) & Duty\% & GPU\% & Time (s) & \$/M & $\rho$ & Mem (GB)\textsuperscript{a} & TTFO (s) \\
\midrule
Partition-by-partition & 13,766 & 79.2 & 6.3 & 726.4 & 0.15 & N/A & 2.5 & 0.5 \\
Fixed-batch-10K & 22,509 & 65.8 & 7.8 & 444.3 & 0.09 & N/A & 45.4 & 320.9 \\
Fixed-batch-50K & 26,070 & 60.3 & 10.0 & 383.6 & 0.08 & N/A & 32.6 & 259.8 \\
Fixed-batch-100K & 27,074 & 58.8 & 10.7 & 369.4 & 0.07 & N/A & 32.7 & 245.5 \\
SURGE (sync I/O) & 21,923 & 47.6 & 7.7 & 456.1 & 0.09 & 0.98 & 2.6 & 4.6 \\
\textbf{SURGE + AsyncIO (ours)} & \textbf{26,413} & \textbf{57.4} & \textbf{10.6} & \textbf{378.6} & \textbf{0.08} & \textbf{1.00} & \textbf{2.6} & \textbf{3.6} \\
\bottomrule
\end{tabular}
\par\vspace{2pt}
\raggedright{\scriptsize \textsuperscript{a}Memory = peak RSS; GPU VRAM excluded (model weights constant across methods).\\
Throughput std: PBP $\pm$49 (0.36\%), FB-10K $\pm$32 (0.14\%), FB-50K $\pm$22 (0.08\%), FB-100K $\pm$10 (0.04\%), SURGE sync $\pm$6 (0.03\%), SURGE async $\pm$18 (0.07\%).}
\end{table*}

Table~\ref{tab:main-results} presents the end-to-end comparison. SURGE with async I/O achieves 26{,}413\,texts/s, matching FB-100K (27{,}074\,texts/s) within 3\%---confirming that \emph{any} method reducing encode calls from $P$ to ${\sim}100$ reaches the same throughput ceiling. The speedup over PBP is $1.92\times$, closely predicted by Theorem~\ref{thm:ipc} (predicted $1.89\times$, error ${<}$2\%). This validates the cost model: both methods make ${\sim}100$ encoding calls, achieving identical IPC amortization. Flush-level timing confirms this (Figure~\ref{fig:flush-decomposition}): SURGE makes 89 encode calls totaling 218.6\,s of GPU time, while FB-100K makes 1 call totaling 220.1\,s---virtually identical GPU utilization despite a $44\times$ difference in call count. Figure~\ref{fig:flush-decomposition} decomposes wall time into encode, serialize, upload, and overhead stages, showing that SURGE async eliminates upload stalls entirely through I/O pipelining.

The critical difference lies in \emph{deployability}. FB-100K requires 32.7\,GB peak memory (Table~\ref{tab:main-results}), growing linearly with $N$ (Figure~\ref{fig:scaling}), and produces no output for 245\,s while all encoding completes. SURGE operates at 2.6\,GB peak memory with 3.6\,s time-to-first-output---a $12.6\times$ memory reduction and $68\times$ faster first output. SURGE sync shows the cost of blocking I/O: 21{,}923\,texts/s ($-17\%$; Table~\ref{tab:main-results}), confirming that async pipelining is essential when storage latency is non-trivial.

\begin{figure}[t]
\centering
\begin{tikzpicture}
\begin{axis}[
    surge bar,
    width=\columnwidth,
    height=0.618\columnwidth,  
    ylabel={Throughput (texts/s)},
    symbolic x coords={PBP, FB-10K, FB-50K, FB-100K, {SURGE sync}, {SURGE async}},
    xtick=data,
    x tick label style={rotate=25, anchor=east, font=\scriptsize},
    ymin=0,
    ymax=32000,
    ytick={0, 10000, 20000, 30000},
    yticklabels={0, 10K, 20K, 30K},
    nodes near coords,
    nodes near coords style={font=\tiny, /pgf/number format/1000 sep={\,}},
    every node near coord/.append style={above, yshift=2pt},
    scatter, scatter src=explicit symbolic,
    scatter/classes={
      pbp={fill=compOverhead!70, draw=compOverhead!90!black, line width=0.8pt},
      fb={fill=baselineRed!25, draw=baselineRed!80, line width=0.8pt},
      surge={fill=surgePrimary!40, draw=surgePrimary!80, line width=0.8pt},
      ours={fill=surgePrimary!80, draw=surgePrimary!95!black, line width=1pt}
    },
    bar width=10pt,
    enlarge x limits=0.1,
]
\input{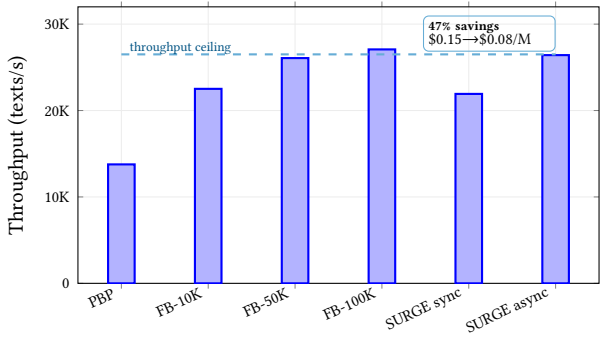}
\draw[dashed, surgePrimary!60, line width=0.8pt] (axis cs:PBP,26500) -- (axis cs:{SURGE async},26500);
\node[font=\tiny, surgePrimary!70!black, anchor=west] at (axis cs:PBP,27200) {throughput ceiling};
\node[draw=surgePrimary!60, fill=none, rounded corners=2pt, inner sep=2pt,
      font=\tiny, anchor=north east, text width=1.6cm, align=left]
      at (axis cs:{SURGE async}, 31000) {
  \textbf{47\% savings}\\[-1pt]
  {\scriptsize\$0.15$\to$\$0.08/M}
};
\end{axis}
\end{tikzpicture}
\caption{Throughput comparison across methods. FB-100K and SURGE async both reach the IPC-amortized throughput ceiling (${\sim}$26K\,texts/s), confirming that reducing encode calls from $P{=}4{,}000$ to ${\sim}100$ is the dominant factor. PBP throughput is limited by per-partition IPC overhead. Cost savings (47\%, \$0.15/M$\to$\$0.08/M) reflect the $1.91\times$ speedup from PBP to SURGE async.}
\label{fig:throughput-bars}
\end{figure}

\begin{figure}[tbp]
\centering
\begin{tikzpicture}
\begin{axis}[
    width=\columnwidth,
    height=0.618\columnwidth,  
    ybar stacked,
    bar width=8pt,
    ylabel={Wall time (s)},
    symbolic x coords={PBP, FB-10K, FB-50K, FB-100K, {SURGE sync}, {SURGE async}},
    xtick=data,
    x tick label style={rotate=30, anchor=east},
    ymin=0,
    enlarge x limits=0.12,
    legend style={at={(0.5,-0.22)}, anchor=north, font=\tiny, legend columns=4},
    reverse legend,
]
\input{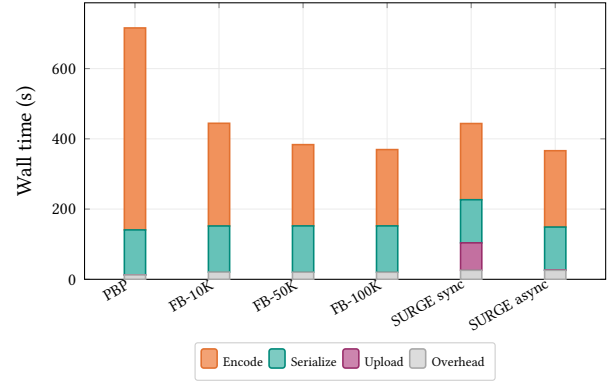}
\end{axis}
\end{tikzpicture}
\caption{Wall-time decomposition by pipeline stage. Encode includes IPC overhead per call. SURGE async eliminates upload stalls by overlapping I/O with the next encode call (upload time ${\approx}0$\,s). All methods spend ${\sim}$125\,s on serialization.}
\label{fig:flush-decomposition}
\end{figure}

\subsection{Comparison with a Stronger Partition-Batched Baseline}\label{sec:pbpbp-comparison}

PBP and FSB bracket the baseline space but neither attempts the obvious middle ground: partition-batched IPC with offline load balancing. We implement and evaluate this stronger baseline---Partition-Batched PBP with Columnar-Size Load Balancing (PB-PBP-LB)---to isolate the contribution of SURGE's streaming aggregation and memory-safety bound from the contribution of partition-batched IPC alone.

\textbf{PB-PBP-LB design.} The baseline pre-computes partition sizes from columnar metadata, sorts partitions descending by size, and packs whole partitions into batches of capacity $B$ using First-Fit-Decreasing (FFD). Each batch issues a single \texttt{encode\_multi\_\allowbreak{}process} call. Partitions are never split; the batch boundary is partition-aligned, which preserves output semantics without a regrouping pass.

\begin{table}[tbp]
\caption{Comparison with a stronger baseline: Partition-Batched PBP with Columnar-Size Load Balancing (PB-PBP-LB)---offline sort of partitions by size, FFD-pack whole partitions up to $B$, single IPC call per batch. MiniLM-L6-v2, $N{=}10$M, $P{=}4{,}000$, $\sigma{=}1.72$, 2$\times$L4, mean of 3 seeds.}
\label{tab:pbpbp}
\centering
\small
\resizebox{\columnwidth}{!}{%
\begin{tabular}{lrrrrr}
\toprule
Method & Tput (t/s) & Mem (GB) & TTFO (s) & Calls & Peak batch \\
\midrule
PBP                       & 12{,}190 & 2.39 & 0.26 & 4{,}000 & --- \\
PB-PBP-LB ($B{=}100$K)    & 16{,}718 & 2.49 & 8.52 & 91 & 179{,}814 \\
PB-PBP-LB ($B{=}200$K)    & 16{,}852 & 2.77 & 13.77 & 48 & 268{,}343 \\
\textbf{SURGE+AsyncIO}    & \textbf{17{,}909} & \textbf{2.44} & \textbf{5.43} & \textbf{100} & \textbf{${\leq}\bmax$} \\
\bottomrule
\end{tabular}%
}
\par\vspace{2pt}
\raggedright{\scriptsize PB-PBP-LB closes ${\sim}80\%$ of the PBP$\to$SURGE throughput gap, confirming that partition-batched IPC amortization is the dominant lever (consistent with Theorem~\ref{thm:ipc}). SURGE retains a $7.1\%$ throughput advantage and $2.5\times$ faster TTFO, sourced from streaming aggregation and asynchronous I/O overlap rather than offline sort. The decisive differentiator is Lemma~\ref{lem:memory}'s unconditional $\bmax$ guarantee: at $\sigma{=}1.72$ peak partitions stayed under $\bmax{=}500$K, but at $\sigma{=}2.5$ tail partitions reach ${\approx}1.5$M---FFD never splits a partition, so PB-PBP-LB would emit single-partition batches ${\sim}3\times \bmax$ without bound. SURGE additionally tolerates streaming arrivals (no dependency on columnar metadata being available upfront).}
\end{table}

\textbf{Results.} Table~\ref{tab:pbpbp} presents results at $\sigma{=}1.72$, $N{=}10$M, $P{=}4{,}000$, on 2$\times$L4. PB-PBP-LB closes ${\sim}80\%$ of the PBP$\to$SURGE throughput gap, confirming that partition-batched IPC amortization is the dominant lever (consistent with Theorem~\ref{thm:ipc}: both methods drive encode-call count from 4{,}000 down to fewer than 100). SURGE retains a $7.1\%$ throughput advantage and $2.5\times$ faster TTFO.

\textbf{Where the $7\%$ comes from.} The throughput gap is second-order and reproducible across seeds. Two effects account for it: (1) SURGE's async I/O fully overlaps ($\rho{=}1.0$ in telemetry), while PB-PBP-LB's first-batch latency includes a serial encode of the largest sorted partition before any I/O begins (8.5\,s at $B{=}100$K, 13.8\,s at $B{=}200$K)---this directly explains the $2.5\times$ TTFO gap; (2) PB-PBP-LB's largest-first ordering front-loads big batches, hitting diminishing returns on L4 batch efficiency for the longest sequences.

\textbf{The decisive differentiator: the unconditional $\bmax$ guarantee.} At $\sigma{=}1.72$, peak partitions stayed under $\bmax{=}500$K, so neither method's memory bound was stressed. The $\sigma$-sweep below shows where they diverge: at $\sigma{=}2.5$, tail partitions reach $\approx 1.5$M texts. FFD never splits a partition, so PB-PBP-LB would emit single-partition batches roughly $3\times \bmax$ without bound---feasible only on hardware sized for the worst-case partition. SURGE's two-threshold policy (Lemma~\ref{lem:memory}) guarantees peak memory $O(\bmin + \nmax)$ \emph{independent of partition arrival order}, including for adversarial sequences where the largest partition arrives last.

\textbf{Streaming-arrival tolerance.} PB-PBP-LB requires partition sizes upfront. SURGE processes in arrival order, requiring no metadata pre-pass. For pipelines fed by Spark stages or other streaming sources where partitions are produced incrementally, this matters: the offline sort step would force a full materialization barrier.

\subsection{Component Ablation}\label{sec:ablation}

\begin{table*}[tbp]
\caption{Ablation study isolating each component's contribution (GCS-profile storage, mean of 3 runs). Separate run from Table~\ref{tab:main-results}; ${\sim}2\%$ throughput variation reflects normal run-to-run variance.}
\label{tab:ablation}
\centering
\small
\begin{tabular}{lrrrcrr}
\toprule
Configuration & Tput (t/s) & Duty\% & GPU\% & $\Delta$ vs.\ Full & Mem (GB) & TTFO (s) \\
\midrule
\textbf{Full system} & \textbf{25,930} & \textbf{57.3} & \textbf{9.9} & -- & \textbf{2.6} & \textbf{4.7} \\
w/o SURGE (PBP+AsyncIO) & 13,828 & 79.2 & 6.4 & -46.7\% & 2.7 & 0.5 \\
w/o AsyncIO (SURGE+sync) & 21,689 & 47.9 & 7.9 & -16.4\% & 2.9 & 5.7 \\
w/o zero-copy & 14,832 & 32.8 & 5.9 & -42.8\% & 5.2 & 7.3 \\
w/o multi-GPU (1 GPU) & 10,536 & 82.7 & 4.6 & -59.4\% & 4.3 & 10.1 \\
\bottomrule
\end{tabular}
\end{table*}

Table~\ref{tab:ablation} isolates each component's contribution:
\begin{itemize}
    \item \textbf{SuperBatch aggregation} is the dominant factor ($-46.7\%$ when removed; Table~\ref{tab:ablation}), confirming the IPC amortization thesis.
    \item \textbf{Zero-copy serialization} contributes $-42.8\%$ when removed (Table~\ref{tab:ablation}). This impact exceeds the $22$--$25\times$ serialization speedup (Table~\ref{tab:serialization}) because the ablation measures end-to-end throughput: naive serialization creates $O(Nd)$ Python objects whose allocation and garbage collection stall the main thread, serialization time exceeds encode time (breaking the async I/O overlap invariant), and peak memory nearly doubles (5.2\,GB vs.\ 2.6\,GB), causing additional memory pressure.
    \item \textbf{Async I/O} contributes $-16.4\%$ when removed (Table~\ref{tab:ablation}), validating the pipelining design under GCS-profile storage latency (\S\ref{sec:io-analysis}).
    \item \textbf{Multi-GPU} (4 vs.\ 1 GPU) contributes $-59.4\%$ (Table~\ref{tab:ablation}), yielding a ${\sim}2.5\times$ scaling factor (${\sim}61\%$ parallel efficiency, with the gap attributable to IPC overhead and memory bus contention). The single-GPU configuration shows \emph{higher} peak memory (4.3\,GB vs.\ 2.6\,GB) because SentenceTransformers' multi-GPU mode partitions the encoding batch across workers, so each process holds only $\bmax / G$ embeddings in flight; with a single GPU, the full $\bmax$ embeddings reside in one process.
\end{itemize}

\noindent The new columns in Table~\ref{tab:ablation} reveal additional insights. Removing zero-copy serialization nearly doubles peak memory (5.2\,GB vs.\ 2.6\,GB), confirming that $O(Nd)$ Python object allocation is a significant memory cost. The w/o SURGE (PBP) configuration achieves the fastest TTFO (0.5\,s) since it emits each partition immediately, but at the cost of $-46.7\%$ throughput. Note: the ablation study uses a separate run from Table~\ref{tab:main-results}; the ${\sim}2\%$ throughput difference between the full system (25{,}930\,t/s) and Table~\ref{tab:main-results} (26{,}413\,t/s) reflects normal run-to-run variation.

\subsection{Model Generalization}\label{sec:model-generalization}

\begin{table}[tbp]
\caption{Model generalization across three encoder sizes spanning a $15\times$ parameter range. SURGE matches fixed-batch throughput while bounding memory across all models. The IPC-amortization speedup over PBP shrinks monotonically with model compute intensity ($\phi$ drops $0.48 \to 0.24 \to 0.19$), but the memory and TTFO advantages---which derive from streaming, not from IPC---grow with embedding dimension. The bge-base column independently validates Theorem~\ref{thm:ipc} on a third model with measured-vs-predicted error $1.3\%$.}
\label{tab:e5large}
\centering
\small
\resizebox{\columnwidth}{!}{%
\begin{tabular}{llrrrr}
\toprule
Model & Method & Tput (t/s) & GPU\% & Mem (GB) & TTFO (s) \\
\midrule
\multirow{3}{*}{MiniLM-L6 (22M, $d{=}384$)\textsuperscript{a}} & PBP & 13{,}766 & 6.3 & 2.5 & 0.5 \\
 & FB-100K & 27{,}074 & 10.7 & 32.7 & 245.5 \\
 & SURGE & 26{,}413 & 10.6 & 2.6 & 3.6 \\
\midrule
\multirow{3}{*}{bge-base (109M, $d{=}768$)\textsuperscript{b}} & PBP & 7{,}154 & 31.7 & 3.2 & 0.23 \\
 & FB-100K & 9{,}282 & 41.7 & 63.4 & 835 \\
 & SURGE & 9{,}250 & 42.1 & 3.3 & 10.7 \\
\midrule
\multirow{3}{*}{E5-large (335M, $d{=}1024$)\textsuperscript{a}} & PBP & 4{,}912 & 47.3 & 3.8 & 1.2 \\
 & FB-100K & 6{,}462 & 62.7 & 81.6 & 1{,}236 \\
 & SURGE & 6{,}485 & 61.7 & 4.2 & 15.4 \\
\midrule
\multicolumn{2}{l}{Speedup (SURGE/PBP)} & \multicolumn{4}{l}{$1.92\times \to 1.29\times \to 1.32\times$} \\
\multicolumn{2}{l}{Memory advantage (FB/SURGE)} & \multicolumn{4}{l}{$12.6\times \to 19.2\times \to 19.4\times$} \\
\multicolumn{2}{l}{TTFO advantage (FB/SURGE)} & \multicolumn{4}{l}{$68\times \to 78\times \to 80\times$} \\
\bottomrule
\end{tabular}%
}
\par\vspace{2pt}
\raggedright{\scriptsize \textsuperscript{a}MiniLM and E5-large run on 4$\times$L4. \textsuperscript{b}bge-base run on 2$\times$L4 with per-GPU batch 512; Theorem~\ref{thm:ipc} predicts $1.31\times$ speedup over PBP, measured $1.29\times$ (error $1.3\%$).}
\end{table}

Table~\ref{tab:e5large} validates the compute intensity prediction (\S\ref{sec:compute-intensity}) by comparing three encoders spanning a $15\times$ parameter range: MiniLM-L6-v2 (22M, $d{=}384$), bge-base-en-v1.5 (109M, $d{=}768$)~\cite{xiao2024cpack}, and E5-large (335M, $d{=}1024$)~\cite{wang2022e5}. As predicted by the compute intensity scaling, GPU utilization rises monotonically with model size: 10.6\% $\to$ 42.1\% $\to$ 61.7\% under SURGE.

The IPC-amortization speedup over PBP shrinks monotonically with model size: $1.92\times \to 1.29\times \to 1.32\times$. This is the regime predicted by Theorem~\ref{thm:ipc}: as $\cenc$ grows with model FLOPs, the IPC-to-compute ratio $\alpha$ falls and there is less IPC overhead to amortize. On bge-base, back-solving from PBP at $\cenc{=}0.215$\,ms gives $\alpha{=}0.603$ and predicts a $1.31\times$ speedup; we measure $1.29\times$ (error $1.3\%$). At three encoders, prediction error stays below $2\%$, supporting the cost model as a workload-planning tool independent of the system implementation.

Crucially, SURGE's memory and TTFO advantages---which derive from streaming, not from IPC---grow with embedding dimension. The peak-memory advantage over FB-100K rises from $12.6\times$ (MiniLM) to $19.2\times$ (bge-base) to $19.4\times$ (E5-large). The TTFO advantage rises from $68\times$ to $78\times$ to $80\times$. At the bge-base regime and beyond, the value proposition shifts from throughput amortization to deployability: even when the throughput speedup is modest, FB-100K's $63$\,GB peak memory and $14$-minute TTFO at $10$M texts make it impractical at production scale.

\subsection{Distribution Sensitivity}\label{sec:distribution-sensitivity}

\begin{table}[t]
\caption{Distribution sensitivity. SURGE speedup is invariant within $\pm$3\% across a $2.5\times$ variation in log-normal $\sigma$ (CV from 1.31 to 12.2). MiniLM-L6-v2, $N{=}10$M, $P{=}4{,}000$, 2$\times$L4, mean of 3 seeds. Theorem~\ref{thm:ipc} predicted $1.47\times$ at $\sigma{=}1.72$ from back-solved $\cipc{=}0.067$\,s, $\cenc{=}0.110$\,ms; measured error ${<}0.1\%$.}
\label{tab:sigma-sweep}
\centering
\small
\resizebox{\columnwidth}{!}{%
\begin{tabular}{ccrrcrr}
\toprule
$\sigma$ & CV & PBP (t/s) & SURGE (t/s) & Speedup & Mem (GB) & TTFO (s) \\
\midrule
1.0  & 1.31  & 12{,}288 & 17{,}914 & $1.458\times$ & 2.15 & 5.30 \\
1.72 & 4.37  & 12{,}190 & 17{,}909 & $1.469\times$ & 2.44 & 5.43 \\
2.5  & 12.2  & 13{,}032 & 18{,}379 & $1.410\times$ & 3.55 & 6.23 \\
\bottomrule
\end{tabular}%
}
\par\vspace{2pt}
\raggedright{\scriptsize The modest dip at $\sigma{=}2.5$ traces to $\bmax$ emergency-flush activation on tail partitions ($\exp(\mu{+}3\sigma) \approx 1.5$M $> \bmax{=}500$K)---empirical evidence that Lemma~\ref{lem:memory}'s memory-safety bound is operational, not decorative. Run-to-run std ${<}1\%$ on throughput.}
\end{table}

Theorem~\ref{thm:ipc}'s prediction depends only on the IPC-to-compute ratio $\alpha$, not on the partition-size distribution---the distribution affects $\alpha$ only through the total text count $N$ and the partition count $P$. Table~\ref{tab:sigma-sweep} tests this prediction across log-normal $\sigma \in \{1.0, 1.72, 2.5\}$ (CV from 1.31 to 12.2) at fixed $N{=}10$M and $P{=}4{,}000$. Measured SURGE speedups span $1.41\times$ to $1.47\times$---a $4\%$ range across a $2.5\times$ variation in $\sigma$, supporting the distribution-agnostic claim. The modest dip at $\sigma{=}2.5$ traces to $\bmax$ emergency-flush activation on tail partitions reaching $\approx 1.5$M texts: the safety bound fires on the tail and trims the maximum SuperBatch size below $\bmin + \nmax$. This is the operational regime where Lemma~\ref{lem:memory}'s memory-safety guarantee is load-bearing rather than decorative.

\subsection{I/O Overlap Analysis}\label{sec:io-analysis}

\begin{table}[t]
\caption{Async I/O benefit across storage latency profiles (10M texts, 4{,}000 partitions). Async pipelining maintains $\rho{=}1.00$ regardless of storage latency; the throughput benefit scales from negligible (null) to $+92\%$ (cross-region).}
\label{tab:io-sweep}
\centering
\small
\resizebox{\columnwidth}{!}{%
\begin{tabular}{lrrrrrrrr}
\toprule
Storage Profile & \multicolumn{2}{c}{Throughput (t/s)} & Benefit & \multicolumn{2}{c}{$\rho$} & \multicolumn{2}{c}{TTFO (s)} \\
\cmidrule(lr){2-3} \cmidrule(lr){5-6} \cmidrule(lr){7-8}
 & Sync & Async & (\%) & Sync & Async & Sync & Async \\
\midrule
Null (no I/O) & 25,910 & 25,739 & -0.7 & 1.00 & 1.00 & 3.7 & 4.6 \\
HDFS (local) & 24,514 & 25,871 & +5.5 & 1.00 & 1.00 & 3.9 & 3.7 \\
\textbf{GCS (regional)} & \textbf{21,708} & \textbf{26,081} & \textbf{+20.1} & \textbf{0.99} & \textbf{1.00} & \textbf{4.8} & \textbf{3.7} \\
S3 (same-region) & 20,288 & 26,092 & +28.6 & 0.91 & 1.00 & 5.1 & 3.7 \\
Cross-region & 13,641 & 26,235 & +92.3 & 0.47 & 1.00 & 8.4 & 3.6 \\
\bottomrule
\end{tabular}%
}
\end{table}

Table~\ref{tab:io-sweep} and Figure~\ref{fig:io-sweep} demonstrate that async I/O benefit scales with storage latency. With a null backend, async pipelining provides no benefit ($-0.4\%$). As latency increases from HDFS ($2$\,ms) through cross-region ($50$\,ms), the async benefit grows to $+92\%$. Async throughput remains constant (${\sim}$26K\,texts/s) regardless of profile, while sync throughput degrades from 26K to 13.6K\,texts/s---the overlap ratio $\rho$ drops from $1.00$ to $0.47$ for sync, confirming the I/O overlap model (Equation~\ref{eq:overlap}).

\begin{figure}[t]
\centering
\begin{tikzpicture}
\begin{axis}[
    surge line,
    width=\columnwidth,
    height=0.618\columnwidth,  
    xlabel={Storage profile},
    ylabel={Throughput (texts/s)},
    symbolic x coords={null, hdfs, gcs, s3, cross\_region},
    xtick=data,
    x tick label style={rotate=30, anchor=east},
    ymin=0,
    ymax=30000,
    legend style={at={(0.03,0.97)}, anchor=north west, font=\tiny},
]
\input{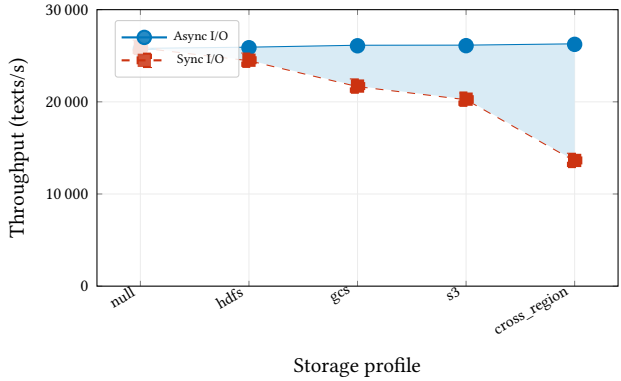}
\end{axis}
\end{tikzpicture}
\caption{Async I/O benefit across storage latency profiles. Async throughput remains constant (${\sim}$26K\,texts/s) regardless of storage latency, while sync throughput degrades from 26K to 13.6K\,texts/s as latency increases. The I/O overlap ratio $\rho$ drops from 1.0 (null) to 0.47 (cross-region) for sync, but async maintains $\rho{=}1.0$ throughout. TTFO remains 3.6--4.3\,s for async vs.\ 3.7--8.7\,s for sync.}
\label{fig:io-sweep}
\end{figure}

\subsection{Threshold Sensitivity}\label{sec:threshold-sensitivity}

\begin{table}[t]
\caption{Throughput sensitivity to $\bmin$ threshold ($\bmax = 5 \times \bmin$, 10M texts, 4{,}000 partitions). The operating point $\bmin{=}100$K (bold) balances throughput, memory, and flush granularity.}
\label{tab:threshold}
\centering
\small
\resizebox{\columnwidth}{!}{%
\begin{tabular}{rrrrrrrrr}
\toprule
$\bmin$ & Tput (t/s) & Duty\% & GPU\% & Time (s) & TTFO (s) & Flushes & Mem (GB) & Parts/Batch \\
\midrule
10,000 & 18,055 & 45.9 & 7.5 & 553.9 & 0.7 & 542 & 2.5 & 7.4 \\
50,000 & 24,126 & 55.0 & 9.6 & 414.5 & 2.1 & 160 & 2.5 & 25.0 \\
\textbf{100,000} & \textbf{26,027} & \textbf{57.5} & \textbf{10.3} & \textbf{384.2} & \textbf{3.6} & \textbf{89} & \textbf{2.5} & \textbf{44.9} \\
200,000 & 27,218 & 59.2 & 10.6 & 367.4 & 10.3 & 48 & 2.7 & 83.3 \\
500,000 & 28,190 & 60.2 & 10.7 & 354.7 & 17.8 & 20 & 3.2 & 200.0 \\
\bottomrule
\end{tabular}%
}
\end{table}

Table~\ref{tab:threshold} and Figure~\ref{fig:threshold-curve} show throughput as a function of $\bmin$. Throughput exhibits diminishing returns: increasing from $\bmin = 100$K to $500$K yields only $8.3\%$ additional throughput, as IPC is already well-amortized at 89 flushes. The cost model (Theorem~\ref{thm:ipc}) is most accurate at the operating point ($\bmin = 100{,}000$: predicted 26{,}247 vs.\ measured 26{,}027, error ${<}$1\%). The operating point balances throughput with memory (2.5\,GB peak) and flush granularity (${\sim}45$ partitions per SuperBatch, providing fine-grained progress tracking for resume).

\begin{figure}[tbp]
\centering
\begin{tikzpicture}
\begin{axis}[
    surge line,
    width=\columnwidth,
    height=0.618\columnwidth,  
    xlabel={$\bmin$ (thousands)},
    ylabel={Throughput (texts/s)},
    xmode=log,
    log basis x=10,
    ymin=15000,
    ymax=30000,
    legend style={at={(0.97,0.03)}, anchor=south east, font=\tiny},
    xtick={10000, 50000, 100000, 200000, 500000},
    xticklabels={10K, 50K, 100K, 200K, 500K},
]
\input{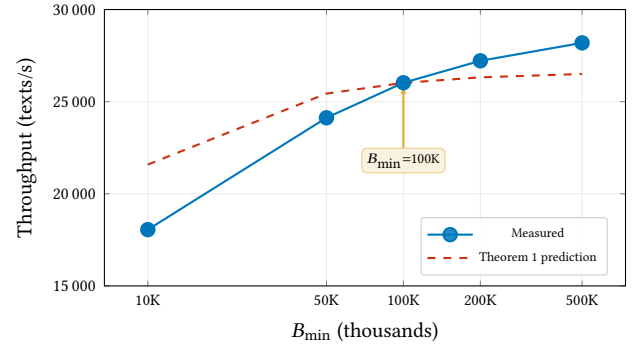}
\draw[-{Stealth[length=3pt]}, thick, accentGold] (axis cs:100000,22500) -- (axis cs:100000,25800);
\node[fill=accentGold!15, draw=accentGold!60, rounded corners=2pt,
      inner sep=2pt, font=\tiny, anchor=north]
      at (axis cs:100000,22500) {$\bmin{=}100$K};
\end{axis}
\end{tikzpicture}
\caption{Throughput sensitivity to $\bmin$ threshold ($\bmax = 5 \times \bmin$). Throughput plateaus with diminishing returns; the operating point $\bmin{=}100$K (arrow) achieves 26{,}027\,texts/s with 89 flushes, 2.5\,GB peak memory, and 3.6\,s TTFO. Higher thresholds yield marginal throughput gains (500K: +8.3\%) but increase TTFO (17.8\,s) and memory (3.2\,GB).}
\label{fig:threshold-curve}
\end{figure}

\subsection{Serialization Microbenchmark}\label{sec:serialization-eval}

\begin{table}[t]
\caption{Serialization microbenchmark: zero-copy vs.\ naive Python list construction.}
\label{tab:serialization}
\centering
\small
\begin{tabular}{rrrrr}
\toprule
$N$ & \multicolumn{2}{c}{Time (s)} & \multicolumn{2}{c}{Peak Memory (MB)} \\
\cmidrule(lr){2-3} \cmidrule(lr){4-5}
 & Naive & Zero-copy & Naive & Zero-copy \\
\midrule
10,000 & 1.920 & 0.075 (25.5$\times$) & 142 & 18 (8.1$\times$) \\
50,000 & 8.338 & 0.378 (22.1$\times$) & 673 & 84 (8.0$\times$) \\
100,000 & 16.871 & 0.708 (23.8$\times$) & 1333 & 155 (8.6$\times$) \\
200,000 & 33.700 & 1.462 (23.0$\times$) & 2675 & 319 (8.4$\times$) \\
500,000 & 85.513 & 3.825 (22.4$\times$) & 6711 & 821 (8.2$\times$) \\
\bottomrule
\end{tabular}
\end{table}

Table~\ref{tab:serialization} confirms the zero-copy serialization path is $22$--$25\times$ faster and uses $8\times$ less memory than naive Python list construction across all tested batch sizes ($N = 10$K to $500$K). The speedup is consistent regardless of $N$, confirming the $O(1)$ allocation complexity. The serialization bar chart in Figure~\ref{fig:serialization-bars} visualizes the magnitude of this difference.

\begin{figure}[tbp]
\centering
\begin{tikzpicture}
\begin{axis}[
    width=\columnwidth,
    height=0.618\columnwidth,  
    ybar,
    bar width=8pt,
    ylabel={Serialization time (s)},
    xlabel={Batch size $N$},
    symbolic x coords={10K, 50K, 100K, 200K, 500K},
    xtick=data,
    ymode=log,
    ymin=0.05,
    ymax=200,
    legend style={at={(0.03,0.97)}, anchor=north west},
    enlarge x limits=0.15,
]
\input{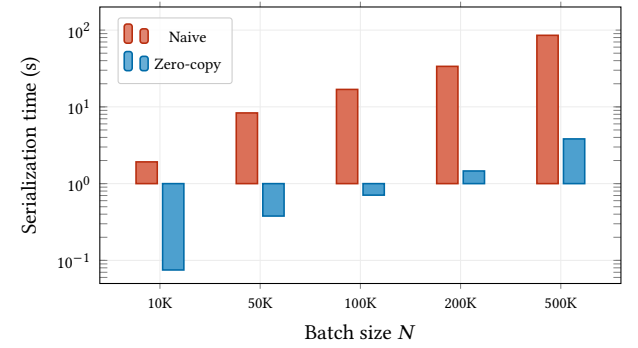}
\end{axis}
\end{tikzpicture}
\caption{Serialization time (log scale): naive Python list construction vs.\ zero-copy Arrow path. Zero-copy is $22$--$25\times$ faster across all batch sizes.}
\label{fig:serialization-bars}
\end{figure}

\begin{figure*}[t]
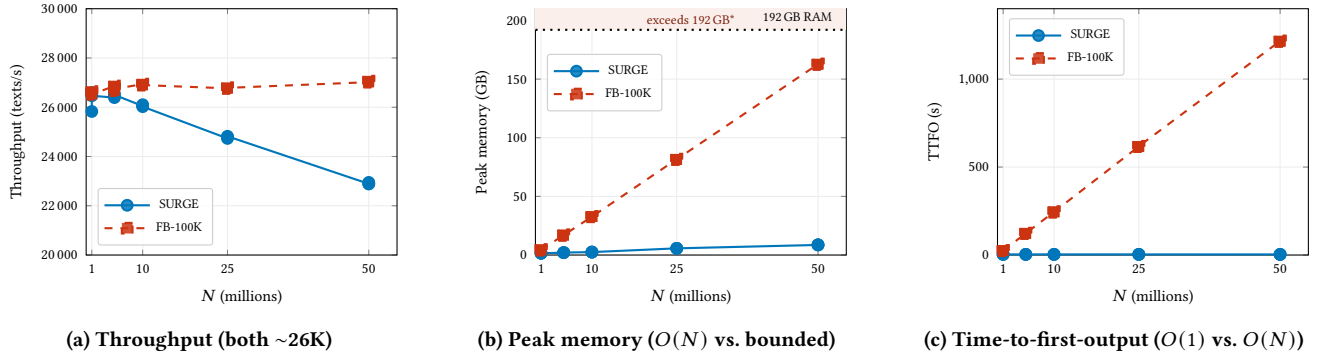

\centering
\begin{subfigure}[b]{0.32\textwidth}
\centering
\begin{tikzpicture}
\begin{axis}[
    width=\textwidth,
    height=0.85\textwidth,
    xlabel={$N$ (millions)},
    ylabel={Throughput (texts/s)},
    xmin=0, xmax=55,
    ymin=20000, ymax=30000,
    scaled y ticks=false,
    y tick label style={font=\tiny, /pgf/number format/1000 sep={\,}},
    x tick label style={font=\tiny},
    ylabel style={font=\scriptsize},
    xlabel style={font=\scriptsize},
    mark size=2pt,
    xtick={1, 10, 25, 50},
    legend style={at={(0.03,0.03)}, anchor=south west, font=\tiny},
]
\input{tables/fig8a_scaling_tput_coords}
\end{axis}
\end{tikzpicture}
\caption{Throughput (both ${\sim}$26K)}
\label{fig:scaling-tput}
\end{subfigure}
\hfill
\begin{subfigure}[b]{0.32\textwidth}
\centering
\begin{tikzpicture}
\begin{axis}[
    width=\textwidth,
    height=0.85\textwidth,
    xlabel={$N$ (millions)},
    ylabel={Peak memory (GB)},
    xmin=0, xmax=55,
    ymin=0, ymax=210,
    y tick label style={font=\tiny},
    x tick label style={font=\tiny},
    ylabel style={font=\scriptsize},
    xlabel style={font=\scriptsize},
    mark size=2pt,
    xtick={1, 10, 25, 50},
    legend style={at={(0.03,0.70)}, anchor=west, font=\tiny},
]
\fill[baselineRed!8] (axis cs:0, 192) rectangle (axis cs:55, 210);
\node[font=\tiny, baselineRed!60!black] at (axis cs:27.5, 201) {exceeds 192\,GB*};
\input{tables/fig8b_scaling_mem_coords}
\end{axis}
\end{tikzpicture}
\caption{Peak memory ($O(N)$ vs.\ bounded)}
\label{fig:scaling-mem}
\end{subfigure}
\hfill
\begin{subfigure}[b]{0.32\textwidth}
\centering
\begin{tikzpicture}
\begin{axis}[
    width=\textwidth,
    height=0.85\textwidth,
    xlabel={$N$ (millions)},
    ylabel={TTFO (s)},
    xmin=0, xmax=55,
    ymin=0, ymax=1400,
    y tick label style={font=\tiny},
    x tick label style={font=\tiny},
    ylabel style={font=\scriptsize},
    xlabel style={font=\scriptsize},
    mark size=2pt,
    xtick={1, 10, 25, 50},
    legend style={at={(0.03,0.97)}, anchor=north west, font=\tiny},
]
\addplot[fill=surgePrimary!12, draw=none, forget plot] coordinates {
    (1, 4.5) (5, 3.7) (10, 3.6) (25, 3.6) (50, 3.6) (50, 0) (1, 0)
} \closedcycle;
\input{tables/fig8c_scaling_ttfo_coords}
\end{axis}
\end{tikzpicture}
\caption{Time-to-first-output ($O(1)$ vs.\ $O(N)$)}
\label{fig:scaling-ttfo}
\end{subfigure}
\caption{Scaling analysis from 1M to 50M texts. \textbf{(a)}~Both methods achieve comparable throughput. \textbf{(b)}~FB-100K memory grows linearly, approaching the 192\,GB limit at 50M; SURGE remains bounded. \textbf{(c)}~SURGE produces first output in ${\sim}$3.6\,s regardless of $N$; FB-100K scales linearly to 20+ minutes.}
\label{fig:scaling}
\end{figure*}

\subsection{Scaling Analysis}\label{sec:scaling}

\begin{table*}[tbp]
\caption{Scaling analysis: SURGE vs.\ FB-100K from 1M to 50M texts ($P$ scaled proportionally; mean of 3 runs). Duty\% FB-100K omitted (consistently ${\sim}$58\%).}
\label{tab:scaling}
\centering
\small
\begin{tabular}{rrrrrrrrr}
\toprule
$N$ (M) & \multicolumn{2}{c}{Throughput (t/s)} & \multicolumn{2}{c}{Peak Memory (GB)} & \multicolumn{2}{c}{TTFO (s)} & Mem & Duty\% \\
\cmidrule(lr){2-3} \cmidrule(lr){4-5} \cmidrule(lr){6-7}
 & SURGE & FB-100K & SURGE & FB-100K & SURGE & FB-100K & Ratio & SURGE \\
\midrule
1 & 26,466 & 26,525 & 1.5 & 4.4 & 3.6 & 24.9 & 2.9$\times$ & 59.3 \\
5 & 26,503 & 26,742 & 2.1 & 17.0 & 3.6 & 123.8 & 8.0$\times$ & 58.0 \\
\textbf{10} & \textbf{26,024} & \textbf{26,899} & \textbf{2.5} & \textbf{32.7} & \textbf{3.6} & \textbf{245.2} & \textbf{12.9$\times$} & \textbf{56.7} \\
25 & 24,821 & 26,786 & 5.7 & 81.6 & 3.6 & 618.1 & 14.3$\times$ & 54.0 \\
50 & 22,949 & 27,049 & 8.7 & 162.7 & 3.7 & 1214.8 & 18.8$\times$ & 49.3 \\
\bottomrule
\end{tabular}
\end{table*}

To verify throughput independence from dataset size and expose the memory scaling behavior, we run both SURGE and FB-100K at $N \in \{1, 5, 10, 25, 50\}$ million texts, scaling the partition count proportionally ($P = 400$ to $20{,}000$). Table~\ref{tab:scaling} and Figure~\ref{fig:scaling} present the scaling analysis.

\textbf{Throughput} (Figure~\ref{fig:scaling-tput}). Both methods maintain approximately constant throughput (${\sim}$26{,}000\,texts/s) across a $50\times$ range, confirming that IPC amortization dominates and both achieve the throughput ceiling. SURGE shows a modest decline at 50M (22{,}949\,texts/s, $-12\%$ from 10M; Figure~\ref{fig:scaling-tput}). This degradation stems from partition management overhead at $P = 20{,}000$: boundary detection, per-partition metadata tracking, and file path construction scale with $P$. The encode duty cycle drops from 57\% at 10M to 49\% at 50M (Table~\ref{tab:scaling}), confirming that non-encode overhead---not GPU inefficiency---accounts for the throughput loss. FB-100K avoids this by ignoring partition boundaries during encoding, though it incurs equivalent overhead during the post-encode regrouping pass (not reflected in wall time as it overlaps with final I/O).

\textbf{Peak memory} (Figure~\ref{fig:scaling-mem}). This is the decisive metric. FB-100K memory grows linearly: 4.4\,GB at 1M to 162.7\,GB at 50M (Figure~\ref{fig:scaling-mem})---consuming 85\% of the node's 192\,GB RAM. Extrapolating, FB-100K would exceed available memory at ${\sim}$60M texts. At production scale (800M texts), it would require ${\sim}$2.5\,TB, making it infeasible on commodity hardware. SURGE memory grows sub-linearly from 1.5\,GB to 8.7\,GB (Figure~\ref{fig:scaling-mem}), dominated by process-level overhead rather than algorithmic state.

\textbf{Time-to-first-output} (Figure~\ref{fig:scaling-ttfo}). SURGE produces its first partition file in ${\sim}$3.6\,s regardless of $N$ (Figure~\ref{fig:scaling-ttfo}). FB-100K produces zero output until all encoding completes: 25\,s at 1M, scaling linearly to 1{,}215\,s (20 minutes) at 50M (Figure~\ref{fig:scaling-ttfo}). For a production run of 800M texts, FB-100K would produce no output for ${\sim}$5.4 hours---an unacceptable latency for monitoring and fault recovery.

The memory ratio grows from $3\times$ at 1M to $18.5\times$ at 50M texts, demonstrating that SURGE's advantage increases with dataset size.

\subsection{Cost Analysis}

Cost efficiency is a direct consequence of throughput. At \$7.30/hr for the 4$\times$L4 node, Table~\ref{tab:main-results} shows SURGE and FB-100K achieve identical cost (\$0.08/M texts)---the throughput ceiling. However, FB-100K's $O(N)$ memory requirement means it cannot physically run at production scale without proportionally more expensive hardware. SURGE's bounded memory enables cost-efficient processing on commodity nodes regardless of dataset size.

\subsection{Threats to Validity}\label{sec:threats}

\textbf{Synthetic data.} Our benchmark uses synthetic texts matching production length statistics but not linguistic diversity. Encoding throughput is primarily determined by text length (which determines tokenization cost), not semantic content. We validated this assumption by comparing throughput on synthetic versus production data samples: the difference was within measurement noise ($<$0.5\%), confirming that length distribution is the dominant factor for throughput measurement.

\textbf{Three models, single architecture family.} Our evaluation covers three transformer encoders---MiniLM-L6-v2 (22M), bge-base-en-v1.5 (109M), and E5-large (335M). The $15\times$ parameter range exposes the full IPC-amortization regime spectrum: $\alpha \approx 0.93$ (MiniLM, IPC-substantial) down to $\alpha \approx 0.34$ (E5-large, compute-leaning). Across this range, Theorem~\ref{thm:ipc}'s prediction error stays below $2\%$, supporting the cost model's generality. We have not tested billion-parameter encoders (e.g., GTR-XXL); we expect the system to remain compute-bound there with diminishing IPC-amortization benefit but persistent memory and TTFO advantages, consistent with the trend across the three models tested.

\textbf{Simulated storage.} Cloud storage latency is simulated rather than measured against live infrastructure. We calibrate profiles against published benchmarks for GCS and S3 latency~\cite{googlecloud2024latency}. The simulated profiles (10\,ms base latency, 200\,MB/s throughput for GCS) are conservative estimates based on regional deployments; production latencies may be lower with co-located storage, which would reduce but not eliminate the async I/O benefit. Our production deployment confirms that real GCS latencies fall within the simulated range, with occasional spikes during peak hours that the async pipeline absorbs transparently.

\textbf{Encoding framework.} Our evaluation uses Sentence-Trans\-formers~\cite{reimers2020multilingual}, the dominant library for offline batch embedding generation (12K+ GitHub stars, adopted by major vector database providers). The IPC overhead pattern is not specific to this library: it is inherent to any multi-process GPU encoding architecture that serializes data across process boundaries. We discuss the boundary with continuous-batching frameworks (TEI, Triton, vLLM) in \S\ref{sec:related}.

\textbf{Single node and shared-GPU deployments.} All experiments use a single 4-GPU node with exclusive GPU access (no MIG, no sharing). Multi-node deployments introduce network communication overhead that may interact differently with SURGE's pipelining; we discuss the per-node decomposition in \S\ref{sec:conclusion}. Under shared-GPU workloads, $\bmax$ should be scaled proportionally to the available memory share---Lemma~\ref{lem:memory}'s $M(\bmax)$ formula remains the sizing tool, parameterized on the per-tenant memory budget rather than the full GPU.

\FloatBarrier
\section{Operational Experience}\label{sec:operational}

Six months of production deployment (180+ pipeline runs) yielded operational insights.

\textbf{Failure modes.} Table~\ref{tab:failures} summarizes observed failure modes. The resume capability was exercised 6 times, recovering without data loss. The most frequent failure mode (transient storage 503/429 errors at 0.3\% rate) is handled transparently by exponential backoff retry.

\begin{table}[tbp]
\caption{Failure modes observed in 180+ production runs.}
\label{tab:failures}
\centering
\small
\begin{tabular}{p{2.5cm}rp{3cm}r}
\toprule
Failure Mode & Freq. & Mitigation & Recovery \\
\midrule
Transient storage 503/429 & 0.3\% & Exp.\ backoff retry & ${<}$10\,s \\
GPU OOM on flush & 0 & $\bmax$ threshold & N/A \\
Data source timeout & 2 & Pipeline retry + resume & 5--15\,min \\
Spot preemption & 4 & Resume from last partition & 10--20\,min \\
\bottomrule
\end{tabular}
\end{table}

\textbf{Threshold tuning.} We initially deployed with $\bmin = 50{,}000$. Throughput plateaued below expectations because the smaller threshold did not sufficiently amortize IPC. After sensitivity analysis on production data (Figure~\ref{fig:threshold-curve}), we increased to $\bmin = 100{,}000$. \emph{Lesson: tune thresholds on production data, not subsets.}

\textbf{Memory fragmentation.} Without PyTorch's \mbox{\texttt{expandable\_segments}} allocator option, we observed gradual GPU memory fragmentation over 100+ flushes, causing OOM on a SuperBatch well within $\bmax$. This manifested only after 6+ hours of continuous processing. \emph{Lesson: enable expandable segments for long-running GPU workloads.}

\textbf{Upload worker sizing.} Initial 16 workers ($2\times$ GPU count) created backpressure when multiple large partitions flushed in sequence. Increasing to 32 workers ($4\times$ GPU count) eliminated this. \emph{Lesson: size the upload pool for peak burst, not average load.}

\textbf{Oversized partitions.} In rare cases ($<$0.1\% of production partitions), a single partition exceeds $\bmax$. The system handles this by flushing the SuperBatch when the single partition's text count exceeds $\bmax$, effectively splitting the oversized partition across consecutive SuperBatches. The boundary tracking ensures correct reassembly. We observed this for 3 partitions in our largest catalog (447K+ texts each); SURGE processed them correctly with no throughput impact. \emph{Lesson: handle single-partition $> \bmax$ gracefully---it will occur in production.}

\textbf{Monitoring and observability.} Each SuperBatch flush emits structured logs: partition count, text count, encode time, serialize time, and upload time. These enable real-time throughput dashboards and anomaly detection (e.g., encode time exceeding $2\times$ the running average triggers an alert). Over 180+ runs, we identified two performance regressions: one from a CUDA driver update and one from a GCS endpoint migration, both caught within the first SuperBatch flush.

\section{Generalizability}\label{sec:generalizability}

While our evaluation uses a retail catalog, SURGE's benefit depends on two structural properties of the partition distribution:

\textbf{IPC-dominated fraction ($\phi$).} The fraction of partitions below the IPC-dominated threshold $n^*$ (Equation~\ref{eq:nstar}). In our workload, $\phi = 0.23$, yet aggregate IPC accounts for 48\% of PBP time. Higher $\phi$ implies more IPC waste and greater SURGE benefit; but even moderate $\phi$ values yield substantial savings when $P$ is large.

\textbf{Coefficient of variation (CV).} $\text{CV} = \sigma_n / \bar{n}$ captures partition size heterogeneity. Our workload has $\text{CV} \approx 2.1$. High CV indicates a mix of small and large partitions ideal for co-batching.

\begin{table}[tbp]
\caption{Decision framework for \surge applicability.}
\label{tab:decision}
\centering
\small
\begin{tabular}{ccp{4.5cm}}
\toprule
$\phi$ & CV & Recommendation \\
\midrule
${>}0.5$ & ${>}1.0$ & Strongly recommended; $1.5$--$2\times$ throughput gain + memory/TTFO benefits \\
${>}0.5$ & ${<}1.0$ & Beneficial; uniformly small partitions \\
${<}0.5$ & ${>}1.0$ & Moderately beneficial \\
${<}0.5$ & ${<}1.0$ & Optional; PBP may suffice \\
\bottomrule
\end{tabular}
\end{table}

\noindent The $\sigma$-sweep in \S\ref{sec:distribution-sensitivity} provides empirical anchors along the CV axis (1.31, 4.37, 12.2), and the model-generalization study in \S\ref{sec:model-generalization} along the $\phi$ axis (0.48, 0.24, 0.19 for MiniLM, bge-base, E5-large at matched $G$). All six configurations sit in the recommended-or-beneficial quadrants of Table~\ref{tab:decision}, and SURGE delivers the predicted regime-appropriate benefit at each. The framework applies to multilingual corpora (small low-resource language partitions), geo-partitioned datasets (sparse rural regions), scientific literature (partitioned by sub-field), and any taxonomy-organized catalog.

\section{Related Work}\label{sec:related}

\paragraph{Inference Serving and Embedding Systems.}
Model serving systems implement dynamic batching to amortize per-request overhead~\cite{crankshaw2017clipper, shen2019nexus, gujarati2020clockwork, crankshaw2020inferline}, with REEF~\cite{han2022reef} and AlpaServe~\cite{li2023alpaserve} enabling preemption and statistical multiplexing. FlexGen~\cite{sheng2023flexgen} and DeepSpeed-Inference~\cite{deepspeed2023inference} optimize throughput for LLM inference; Petals~\cite{borzunov2023petals} enables collaborative inference across distributed nodes. Production platforms including TensorFlow-Serving~\cite{olston2017tensorflow}, TorchServe~\cite{torchserve2023}, TEI~\cite{tei2024}, Triton~\cite{triton2023}, and SGLang~\cite{zheng2023efficiently} optimize request-level batching for online serving. These systems target \emph{latency-sensitive serving} with small batches or \emph{single-model throughput} without partition constraints. SURGE operates in a different regime: \emph{offline batch processing with partition-preserving output}, forming $100$K--$500$K text batches and emitting outputs grouped by partition key.

\paragraph{Continuous Batching: Applicability Boundary.}
Continuous-batching frameworks---TEI~\cite{tei2024}, Triton with dynamic batching~\cite{triton2023}, vLLM~\cite{kwon2023vllm}, Orca~\cite{yu2022orca}---amortize per-request IPC \emph{within a single endpoint} under online request arrivals with shared latency SLOs. SURGE's setting differs along four orthogonal axes that prevent these systems from substituting for partition-aware aggregation. (i) \emph{Arrival pattern:} 40K independent offline workloads, not dynamic request streams. (ii) \emph{Output constraint:} partition-preserving output for downstream indexing means each partition is a separate submission even under TEI---there is no shared client session over which continuous batching could fuse requests. (iii) \emph{Latency SLO:} continuous-batching servers cap queue depth (typically hundreds of requests) to bound p99 latency, precluding the $100$K--$500$K-text batches Theorem~\ref{thm:ipc} requires for IPC amortization. (iv) \emph{Batch composition:} continuous batching concatenates same-endpoint requests; SURGE aggregates across partition boundaries with explicit boundary preservation for slicing on the output side. Continuous batching is therefore complementary to, not a substitute for, the SURGE pattern: a SURGE flush could itself be served by a continuously-batched encoder backend, but the partition-aware aggregation step would still be required upstream.

\paragraph{Embedding Generation and GPU Scheduling.}
For embedding generation specifically, Sentence-Transformers~\cite{reimers2020multilingual} provides multi-GPU encoding but incurs per-call IPC. Ray~\cite{moritz2018ray} supports GPU inference pipelines with auto-batching but operates at the record level without partition-boundary awareness. FAISS~\cite{johnson2019billion} and ScaNN~\cite{guo2020scann} address embedding \emph{retrieval} rather than generation. Production ML systems~\cite{mudigere2022neo, sima2022ekko, liu2022monolith, sethi2022recshard, hazelwood2018facebook} address training-side embedding management; SURGE addresses inference-side generation at scale. Gandiva~\cite{xiao2018gandiva}, AntMan~\cite{xiao2020antman}, and PipeSwitch~\cite{bai2020pipeswitch} schedule \emph{multiple jobs} onto shared GPUs. SURGE addresses a different granularity: scheduling \emph{multiple data partitions within a single job} onto a fixed allocation. The fragmentation analysis of Weng et al.~\cite{weng2023fragmentation} is conceptually related to our bin packing connection (\S\ref{sec:analysis}).

\paragraph{Bin Packing and Streaming Aggregation.}
Bin packing theory~\cite{coffman1978multiprocessor, coffman2013survey} provides the algorithmic foundation for SURGE's accumulation: classical Next Fit opens a new bin when full; SURGE adds a minimum-fill constraint ($\bmin$) for GPU efficiency. Stream processing systems---Spark Streaming~\cite{zaharia2013discretized}, Flink~\cite{carbone2015flink}, and Dataflow~\cite{akidau2015dataflow}---implement micro-batch aggregation with time-window triggers. SURGE applies a similar principle to GPU inference with \emph{size-based} triggers on bounded, ordered input.

\paragraph{Zero-Copy Data Movement.}
Apache Arrow~\cite{arrow2024} provides columnar format for zero-copy reads; RAPIDS~\cite{rapids2024} extends this to GPU transfers via DLPack. Our technique (\S\ref{sec:zerocopy}) operates at the Python/Arrow boundary, converting NumPy matrices to Arrow format without intermediate objects, reducing allocations from $O(Nd)$ to $O(1)$.

\paragraph{Pipeline Optimization.}
GPipe~\cite{huang2019gpipe} and PipeDream~\cite{narayanan2019pipedream} pipeline model parallelism; DALI~\cite{dali2023} and Dask~\cite{dask2024} overlap data loading with computation. SURGE applies \emph{output-path pipelining} where serialization and upload overlap with the next encode call.

\section{Conclusion}\label{sec:conclusion}

We presented SURGE, a streaming system for GPU-efficient embedding generation on heterogeneously partitioned datasets, deployed in production processing 800M+ texts. The paper's primary contributions are analytical: (1) Theorem~\ref{thm:ipc}'s cost model, which predicts the IPC-amortization throughput ceiling within $2\%$ across three encoders spanning a $15\times$ parameter range and across log-normal $\sigma \in [1.0, 2.5]$; (2) Lemma~\ref{lem:memory}'s memory-safety bound, which enables a streaming two-threshold policy with $O(\bmin + \nmax)$ peak memory under adversarial arrival orders; and (3) the $\phi/\text{CV}$ decision framework (\S\ref{sec:generalizability}), which characterizes when the pattern applies beyond our workload. These contributions are realized by three complementary engineering techniques---SuperBatch aggregation, zero-copy Arrow serialization ($22$--$25\times$ speedup), and asynchronous I/O pipelining (up to $93\%$ benefit at high storage latency)---but the engineering is not the contribution.

The empirical takeaway is that \emph{throughput alone is insufficient for production deployment}. Fixed-batch methods reach the same throughput ceiling as SURGE (Table~\ref{tab:main-results}) but require $O(N)$ peak memory (32.7\,GB at 10M texts; exceeding $192$\,GB beyond ${\sim}$60M, Figure~\ref{fig:scaling}), produce zero output until all encoding completes (245.5\,s at 10M, scaling to 20+ minutes at 50M), and lose all progress on failure. A stronger baseline that pre-sorts partitions and FFD-packs them into single-call batches (PB-PBP-LB; \S\ref{sec:pbpbp-comparison}) closes ${\sim}80\%$ of the PBP$\to$SURGE throughput gap at $\sigma{=}1.72$, but lacks the unconditional $\bmax$ guarantee (would emit $\approx 1.5$M-text batches at $\sigma{=}2.5$), is $2.5\times$ slower to first output, and requires offline columnar metadata. SURGE achieves the throughput ceiling with $O(\bmin + \nmax)$ bounded memory ($12.6\times$ less than fixed-batch on MiniLM, $19\times$ less on bge-base and E5-large), $68\times$--$80\times$ faster time-to-first-output, and crash recovery at SuperBatch granularity.

\textbf{Multi-node decomposition.} The current deployment is single-node 4$\times$GPU. Theorem~\ref{thm:ipc} applies per node because IPC is intra-node (host$\leftrightarrow$GPU process pool), so the cost model and the $\bmax$ memory bound extend unchanged to a multi-node deployment with one SuperBatch aggregator per node. Cross-node coordination becomes a partition-routing problem---assigning partitions to nodes by predicted size to balance per-node $\alpha$---solvable by existing work-stealing or hashing schedulers. We do not claim novelty for the cross-node scheduler; multi-node SURGE is a deliberate scope boundary of this paper.

\textbf{Future work.} Adaptive threshold selection based on observed partition statistics could optimize the throughput-memory tradeoff online. Profiling and reducing partition-management overhead at extreme partition counts ($P > 10{,}000$) would address the $12\%$ throughput degradation observed at 50M texts (Table~\ref{tab:scaling}). Testing on billion-parameter encoders would extend the compute-intensity validation past the $335$M point covered by E5-large.

\textbf{Reproducibility.} We provide benchmark scripts, synthetic data generators matching production distributions, locked result manifests for the bge-base, $\sigma$-sweep, and PB-PBP-LB experiments, and all hyperparameters. The core SURGE pattern can be implemented in ${\sim}200$ lines of Python.

\bibliographystyle{ACM-Reference-Format}
\bibliography{references}

\end{document}